\documentclass[12pt]{article}

\usepackage{graphicx,psfrag,epsf}
\usepackage{amsmath}                 
\usepackage{url} 
\usepackage{amsthm}
\usepackage{amssymb}
\usepackage{epstopdf}
\usepackage{ulem}
\usepackage{color}
\usepackage{enumerate}
\usepackage{authblk}
\usepackage{natbib}
\usepackage{xcolor,colortbl}

\usepackage{float}

\usepackage{soul}
\setstcolor{red}

\usepackage{dutchcal}

\theoremstyle{plain}

\newtheorem{definition}{Definition}
\newtheorem{theorem}{Theorem}

\newtheorem{lemma}{Lemma}
\newtheorem{prop}{Proposition}

\newtheorem{remark}{Remark}

\newcommand{\A}{\ensuremath{\mathcal{A}}}
\newcommand{\W}{\ensuremath{\mathcal{W}}}
\newcommand{\bA}{\ensuremath{\bar{A}}}
\newcommand{\bAc}{\ensuremath{\mathcal{\bar{A}}}}
\newcommand{\tA}{\ensuremath{\tilde{A}}}
\newcommand{\tAc}{\ensuremath{\tilde{\A}}}

\newcommand{\T}{\ensuremath{\mathcal{T}}}
\newcommand{\D}{\ensuremath{\mathcal{D}}}
\newcommand{\B}{\ensuremath{\mathcal{B}}}
\newcommand{\C}{\ensuremath{\mathcal{C}}}

\newcommand{\e}{\ensuremath{\epsilon}}
\newcommand{\de}{\ensuremath{\delta}}

\newcommand{\1}{\ensuremath{\mathbf{1}}}
\newcommand{\R}{\ensuremath{\mathbb{R}}}

\newcommand{\half}{\ensuremath{\frac{1}{2}}}

\newcommand\inner[2]{\ensuremath{\langle #1, #2 \rangle}}

\newcommand{\rmE}{\ensuremath{\mathrm{E}}}

\newcommand{\G}{\ensuremath{\mathcal{G}}}
\newcommand{\V}{\ensuremath{\mathcal{V}}}
\newcommand{\E}{\ensuremath{\mathcal{E}}}

\newcommand{\xra}[1]{\mathop{ \xrightarrow{#1} }}

\begin{document} 
\def\spacingset#1{\renewcommand{\baselinestretch}%
{#1}\small\normalsize} \spacingset{1}

\title{\bf Novel Sampling Design for Respondent-driven Sampling}
\author[1]{Mohammad Khabbazian}
\author[2]{Bret Hanlon}
\author[2]{Zoe Russek}
\author[2]{Karl Rohe}
\affil[1]{Department of Electrical and Computer Engineering, University of Wisconsin-Madison}
\affil[2]{Department of Statistics, University of Wisconsin-Madison}

\date{}
\maketitle

\begin{abstract}
	Respondent-driven sampling (RDS) is a method of chain referral sampling
	popular for sampling hidden and/or marginalized populations.  As such,  even
	under the ideal sampling assumptions, the performance of RDS is restricted by
	the underlying social network: if the network is divided into communities
	that are weakly connected to each other, then RDS is likely to oversample one
	of these communities.  In order to diminish the ``referral bottlenecks''
	between communities, we propose anti-cluster RDS (AC-RDS), an adjustment to
	the standard RDS implementation.  Using a standard model in the RDS
	literature, namely, a Markov process on the social network that is indexed by
	a tree, we construct and study the Markov transition matrix for AC-RDS. We
	show that if the underlying network is generated from the Stochastic
	Blockmodel with equal block sizes, then the transition matrix for AC-RDS has
	a larger spectral gap and consequently faster mixing properties than the
	standard random walk model for RDS. In addition, we show that AC-RDS reduces
	the covariance  of the samples in the referral tree compared to the standard
	RDS and consequently leads to a smaller variance and design effect.  We
	confirm the effectiveness of the new design using both the Add-Health
	networks and simulated networks. 
\end{abstract}

\noindent%

\vfill

\noindent%
{\it Acknowledgements:} Zoe Russek and Karl Rohe are supported by NSF grant DMS-1309998,
ARO grant W911NF-15-1-0423, and grants from the Graduate School at UW Madison.

\spacingset{1.1} 

\section{Introduction} 
Several public policy and public health programs depend on estimating characteristics
of hard-to-reach or hidden populations (e.g. HIV prevalence among people who
inject drugs). These hard-to-reach populations cannot be sampled with standard
techniques because there is no way to construct a sampling frame.  
\cite{heckathorn1997respondent, heckathorn2002respondent} proposed
respondent-driven sampling (RDS) as a variant  of chain-referral methods,
similar to snowball sampling \citep{goodman1961snowball, handcock2011comment},
for collecting and analyzing data from hard-to-reach populations. Since then, RDS has been
employed in over 460 studies spanning more than 69 countries
\citep{malekinejad2008using, white2015strengthening}.  

RDS encompasses a collection of methods to both sample a population and infer
population characteristics \citep{salganik2012commentary}, referred to as RDS
sampling and RDS inference, respectively.  RDS sampling starts with a few ``seed''
participants chosen by a convenience sample of the target population. Then, the
initial participants are given a few coupons to refer the second wave of
respondents, the second wave refers the third wave, and so on. The participants
receive a dual incentive to (i) take part in the study and (ii) successfully
refer participants. The dual incentive, limited number of coupons, and
without replacement sampling, in theory, help RDS mix more quickly than snowball sampling,
allowing for the potential to penetrate the broad target population and reduce
its dependency on the initial convenience sample. In addition, in some cases, participants
are provided with extra instructions to conduct without replacement
sampling\footnote{``Please make sure that the persons you give the coupons to
are (add your eligibility criteria here) and have not received this coupon from
someone else'' \citep[][p. 330]{RDSmanual}.} and also reach out to different types
of people in the target population\footnote{``If possible, try and give the
coupons to different types of people who you know (e.g.\ different ages,
different levels of income, from different locations in this city)''
\citep[][p. 330]{RDSmanual}.}.

Since Heckathorn's original RDS paper, the statistical literature on RDS has created
several estimators that seek to reduce the bias and estimate confidence
intervals  \citep{heckathorn2011comment}. The most popular RDS estimators are
generalized Horvitz-Thompson type estimators where the inclusion probabilities
are derived from various models of the sampling procedure
\citep{volz2008probability, gile2011improved, gile2011network}.   

RDS performance has been evaluated through simulation studies
\citep{goel2010assessing, gile2010respondent}, empirical studies
\citep{wejnert2008web, wejnert2009empirical, mccreesh2012evaluation}, and theoretical analyses
\citep{goel2009respondent}. The main message of these studies is that (i) RDS
can suffer from bias; (ii) in some cases, the current RDS estimators do not
reduce bias; and, most importantly, (iii) the estimators have higher variance
than what was initially thought 
\citep{goel2009respondent, goel2010assessing, white2012respondent}.  
To help bridge the gap between theory and practice,
\cite{gile2014diagnostics} suggests various diagnostics to examine the
validity of the modeling assumptions.

For the purpose of computing the inclusion probability and designing
estimators, the Markov chain is typically assumed to be the underlying
generative model. However, this model under the standard formulation does not
take into account the without replacement nature of the RDS sampling process
\citep{gile2011improved, gile2011network} or the effect of
preferential recruitment, the tendency of respondents to refer particular friends
\citep{crawford2017identification, mccoy2013improving}. 
As a result, the designed estimators may fail to
provide credible estimations of the target population characteristics.

\cite{goel2009respondent} and \cite{verdery2015network} analytically study the
effects of homophily and community structure on the variance of the
estimator. Homophily, a common property of social networks, is the
tendency of people to establish social ties with others who share common
characteristics such as race, gender, and age. Strong homophily creates
community structure in the social network.  This in turn creates referral
bottlenecks between different groups in the population; the RDS referral chain
can struggle to cross these bottlenecks, failing to quickly explore the network.
In such situations, RDS is sensitive to the initial convenience sample, leading
to biased estimators.  Moreover, the bottlenecks make successive samples
dependent, leading to highly variable estimators. 
\cite{crawford2017identification} gives a rigorous definition of homophily
and preferential recruitment, and shows that it is difficult to precisely
measure these quantities in practice. The results in \cite{treevar}
show that if the strength of this bottleneck crosses a critical threshold, then
the variance of the standard estimator decays slower than $1/n$, where $n$ is
the sample size.
Furthermore, \cite{verdery2016new} proposes a set of data collection
methods, survey questions, and estimators for RDS to estimate clustering
characteristics and draw inferences about topological properties of social
networks. The basic data they propose to collect is about connected and closed
triplets that participants form by their social ties. They also provide some
measure of clustering levels in RDS samples.

To diminish referral bottlenecks, this paper proposes an adjustment to the
current RDS implementation.  Instead of asking participants to refer anyone
from the target population, this paper proposes two basic types of
``anti-cluster referral requests,'' which are described in Figure
\ref{fig:refreq}. These referral requests diminish referral bottlenecks by
producing triples of participants that do not form a triangle, closed triplet,
in the social network. The figure contains two types of such requests. In fact,
as described in Section \ref{sec:acrw}, we propose a procedure that
probabilistically alternates between the two requests.

As compared to alternative methods, anti-cluster requests are more successful
in diminishing referral bottlenecks for three reasons.  First, this approach
preserves privacy by refraining from asking participants to list their
friends in the population. Second, anti-cluster
requests do not require \textit{a priori} knowledge about the nature of the
bottleneck.  For example, the most salient bottleneck could form on race,
gender, neighborhood, or something else. If researchers knew which of these was
most restricting the sampling process, then perhaps specific requests could be
formed.  However, in many populations, the bottlenecks are not known in
advance.  The final advantage is that the proposed adjustment is mathematically
tractable; under certain assumptions, anti-cluster requests can form a
reversible Markov chain.  

\begin{figure}[h]
	\textbf{Anti-cluster referral requests}
  \centering
  \includegraphics[width=6in]{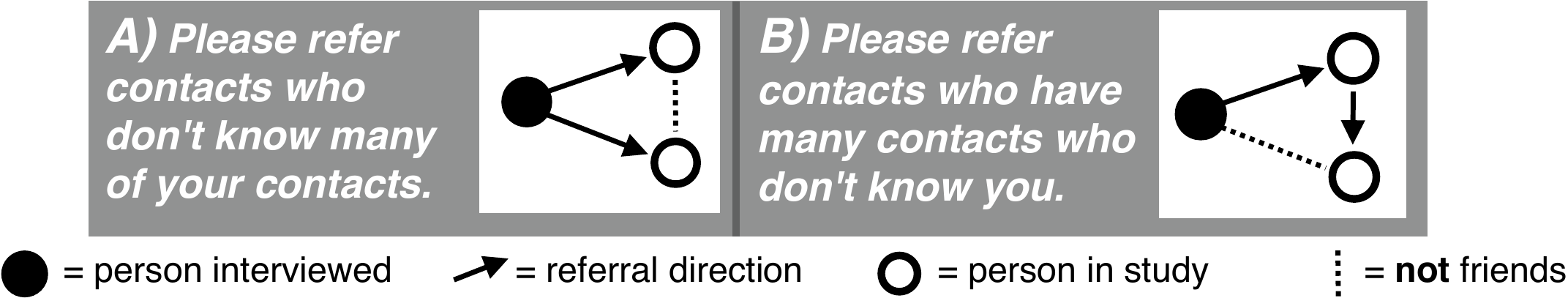}
	\caption{An illustration of two anti-cluster referral requests. 
	The referral requests for anti-cluster sampling are privacy
	preserving because they do not require participants to list all of their
	friends.  Moreover, these requests do not require any knowledge about the
	community structures in the social network.}
  \label{fig:refreq}
\end{figure}

We propose a novel variant of RDS, then study its theoretical properties under
a statistical model. This work provides theoretical motivation to further
develop and study novel referral requests.  Additional work is needed before
this variant should be employed in the field; this is discussed further in
Section \ref{sec:implementation}.

The remainder of the paper is organized as follows. Section
\ref{sec:sampDesign} describes Designed RDS and presents our
proposed design, anti-cluster RDS (AC-RDS). Section
\ref{sec:prelim} sets the notation and provides the
mathematical preliminaries. Section \ref{sec:theory} gives our
theoretical results, distinguishing between ``population graph''
and ``sample graph'' results.  Section \ref{sec:sim} contains
numerical experiments which compare the performance of AC-RDS
with standard RDS.  Section \ref{sec:implementation} 
discusses some gaps between the theory and the practice 
of novel referral requests. We summarize the paper and offer a
discussion in Section \ref{sec:final}. All of the proofs are
provided in the online supplementary material.

\section{Novel Sampling Designs} \label{sec:sampDesign}
When preparing to sample a target population with RDS, some aspects can be
controlled by researchers (e.g. how many referral coupons to give each
participant) and others cannot.  In particular, the social network is beyond
the control of researchers. Community structures are an intrinsic part of
social networks \citep{girvan2002community} which, in RDS, lead to referral
bottlenecks. To minimize these bottlenecks, RDS can be altered to make some
referrals more or less likely. This is the essence of novel sampling designs
for respondent-driven sampling.

As a thought experiment, suppose that the population of interest is divided
into two communities, EAST and WEST. Furthermore, assume that people form most
of their friendships within their own community. Under this simple model,
referrals between communities are unlikely, creating a bottleneck.  
Now, suppose that these communities were known before
performing the sample. The researchers could then request referrals from
specific groups (e.g. flip a coin, if heads request WEST and if tails request
EAST).  This does not change the underlying social network, but it does change
the probability of certain referrals.  If participants followed this request,
the referral bottleneck between EAST and WEST would be diminished.  If
$90\%$ of a participant's friends belonged to the same community as the
participant, then the standard approach would obtain a cross-community referral
only $10\%$ of the time.  However, with the coin flip implementation, such a
referral happens $50\%$ of the time.

\cite{Mouw01082012} propose an alternative technique, Network Sampling with Memory (NSM).
In NSM sampling, researchers construct a sampling frame by asking RDS
participants to nominate their friends in the target
population.  This list  is combined with the friend lists from previous
participants to form a sampling frame. In the ``List'' mode of the
sampling process, the next individual to be recruited and interviewed is
selected by sampling with-replacement from the list of nominated members.  In
the ``Search'' mode, to improve the mixing property of the sampling process,
individuals who appeared to be the ``bridge nodes''to the unexplored parts of the
network are identified. Then, randomly a node from friends of the bridge nodes
who have only 1 nomination is selected for the next interview.
In computational
experiments, \cite{Mouw01082012} report a decrease in the design effect,
the ratio of the sampling variance to the sampling variance of simple random
sampling, of this novel approach.  

These two extensions of RDS (i.e. flipping a coin and NSM) are both forms of
Designed RDS; through novel implementations of the sampling process they
adjust the probability of certain referrals, thereby diminishing the referral
bottlenecks.  Unfortunately, the coin flipping example requires prior
information about the social network, which may be unattainable given the
hidden nature of the target population. The NSM approach requires
respondents to reveal partial name and demographic information of their
friends.  Moreover, it asks respondents to refer (recruit) selected individuals
from the list of nominees. When practically implemented in a hidden population,
however, it is not clear if respondents will be willing to provide the requested 
information or refer the selected individual from their list of nominees. 
Furthermore, the referral process may be based more heavily on participants' 
interactions with members of the target population following the survey than
on any plan they make to refer ahead of time.


Anti-cluster RDS is a type of Designed RDS that complements and builds upon
both of these approaches.  The implementation of anti-cluster RDS does not
require  \textit{a priori} information on the communities in the social
network, nor does it require that participants reveal sensitive information
about individuals who have not consented.  Anti-cluster sampling is designed to
place larger referral probabilities on edges belonging to fewer triangles.
There are at least two ways to consider why this strategy circumvents
bottlenecks.  

\begin{enumerate}
  \item Many empirical networks share three properties.  First, the number of
    edges is proportional to the number of nodes (i.e. the network is globally
    sparse).  Second, friends of friends are likely to be friends (i.e. the
    network is locally dense).  Third, shortest path lengths are small (i.e.
    the network has a small diameter); this is also known as the small-world
    phenomenon.  \cite{watts1998collective} shows how a network can satisfy all
    three properties; take a deterministic graph that satisfies the first two
    features (e.g. a triangular tessellation),  then select a few edges at
    random and randomly re-wire these edges to a randomly chosen node.  Notice
    that these ``random edges'' are unlikely to be contained in a triangle.
		So, edges that are not part of triangles are more likely to lead to
		quicker network traverse. Anti-cluster RDS makes referral along that edges
		more probable, and potentially mixes faster and collects more representative
		samples from the target population.

	\item The Markov chain has been a popular model for studying theoretical properties of RDS.
		Under the with-replacement sampling formulation of this
		model people make referrals by selecting uniformly from their set of
		friends.  A similar assumption could be made about anti-cluster referrals;
		the referral is drawn uniformly from the set of referrals that satisfy the
		anti-cluster request.  If the Markov transition matrix for anti-cluster
		sampling can be shown to have a larger spectral gap than the Markov
		transition matrix for the simple random walk, then this suggests that
		anti-cluster sampling will obtain a more representative sample.  
\end{enumerate}

In this paper, we pursue the second approach. 

\section{Preliminaries} \label{sec:prelim} 
\subsection{Framework} \label{sec:frame}
This paper models the referral process as a Markov chain indexed by a tree
\citep{benjamini1994markov}.
A Markov chain indexed by a tree is a variant of branching Markov chains 
in which a fixed deterministic tree indicates branching.
This model is a straightforward combination of the Markov models developed in
the previous literature on RDS \citep{heckathorn1997respondent, salganik2004sampling, volz2008probability,goel2009respondent}.
This  model is built with the following four mathematical pieces: an underlying social graph,
a node feature which is measured on each sampled node (e.g. HIV status), a
Markov transition matrix on this graph, and a referral tree to index the Markov
process.  Figure~\ref{fig:graphSeq} gives a graphical depiction of this process. 

\begin{figure}[h] 
  \centering
  \includegraphics[width=6.5in]{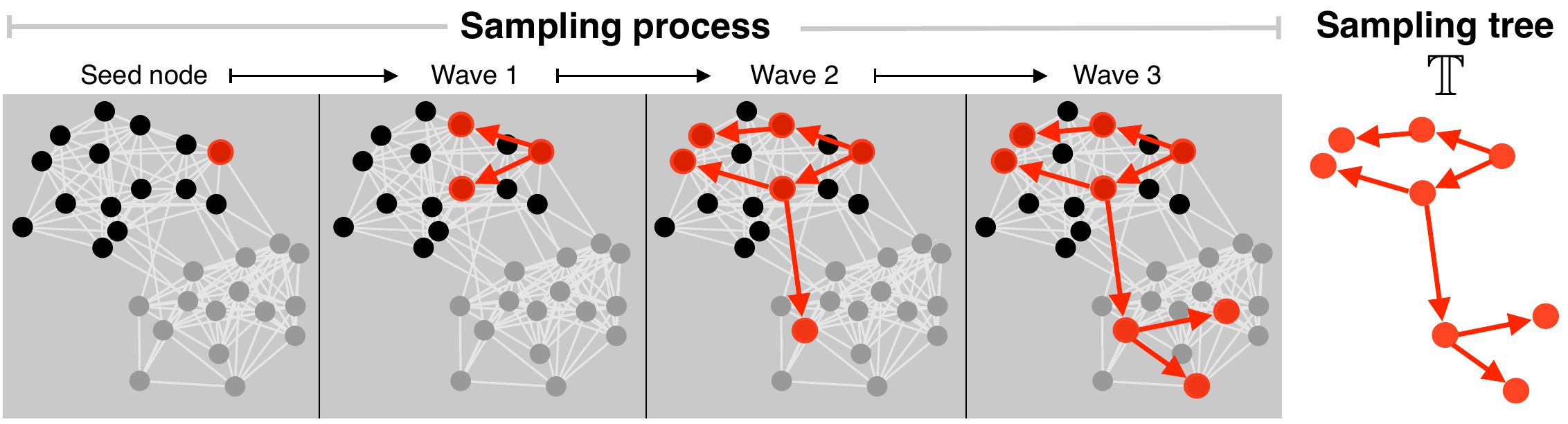} 
  \caption{A graphical depiction of the referral process, which is modeled as a
    Markov chain indexed by a tree. This figure gives an example of a social
    network $\G$ and a referral tree $\mathbb{T}$.}
  \label{fig:graphSeq}
\end{figure}

\textbf{The social network.}
Denote the underlying social network by an undirected graph $\G=(\V,\E)$
where $\V = \{1, \dots, N\}$ is the set individuals in the target
population and $\E = \{(u,v) : \mbox{ $u$ and $v$ are friends}\}$ is the set
of social ties. 
Define the adjacency matrix $A$ as
\begin{equation}
  A(u,v) = \left\{ \begin{array}{ll}
    1 & \mbox{if } (u,v) \in \E;\\
    0 & \mbox{o.w.}\end{array} \right.  
\end{equation}
and the node degree as $\deg(u) = \sum_v A(u,v)$. 

\textbf{Node features.}  
After sampling an individual $u \in \V$, we can measure their status $y(u)$,
where $y:\V\rightarrow \R$ is some node feature. For instance, $y(u)$ could be a
binary variable which is one if node $u$ is HIV+ and zero otherwise. The aim of
RDS is to estimate the population average of $y$ over all nodes, 
\[\mu = \frac 1N \sum_{u\in\V} y(u).\]

\textbf{Markov chain.}    
Let $(X_i)_{i=0}^n$ be an irreducible Markov chain with the finite state space
$\V$ of size $N$ and transition matrix $P \in \R^{N \times N}$;  for $u,v \in \V$
and for all $i \in 0, \dots, n-1$, \[P(u,v) = \Pr(X_{i+1} = v | X_i = u).\]
Define $P_A$ as the Markov transition matrix of the simple random walk, \[P(u,v)
= \frac{A(u,v)}{\deg(u)}.\] The standard Markov model for RDS assumes that $X_i$
is a simple random walk.  

\textbf{Novel designs.} 
Designed RDS is any technique that assigns differing weights to the edges.  Define the mapping
$W: \E \rightarrow R_+$ as a weighting function on the edges $(u,v) \in \E$.  If
$(u, v) \in \E$ and $W(u,v) >0$, then $u$ can recruit $v$.  For simplicity,
define $W(u,v) =0$ if $(u,v) \not \in \E$.  Then, $W$ can be expressed as a
matrix.  Define the diagonal matrix $T$ to contain the row sums of $W$, so that $T_{uu} = \sum_{v}
W(u,v)$.  

Through novel implementations, Designed RDS alters the edge weights.  After
weighting the edges, the Markov transition matrix becomes 
\begin{equation}\label{eq:defp}
P_W = T^{-1} W.
\end{equation} 
If Designed RDS increases an edge weight, it makes the edge more likely to be
traversed.  

We restrict the analysis to symmetric weighting matrices. Because of this
restriction, $P_W$ is reversible and has a stationary distribution $\pi:
\V\rightarrow \R_+$ that is easily computable, 
\begin{equation} \label{eq:stationary}
\pi(u) = \frac{T_{uu}}{\sum_v T_{vv}}.
\end{equation}
 Throughout, it will be assumed that $X_0$ is initialized with
$\pi$.  A more thorough treatment of Markov chains and their stationary
distribution can be found in \cite{levin2009markov}. 

\textbf{Referral tree.} 
In the Markov chain model, participant $X_i$ refers participant $X_{i+1}$. This
assumes that each participant refers exactly one individual. In practice, RDS
participants usually refer between zero and three future participants. To allow for
this heterogeneity, it is necessary to index the Markov process with a tree, not
a chain. Let $\mathbb{T}$ denote a rooted tree with $n$ nodes. See Figure
\ref{fig:graphSeq} for a graphical depiction. 

To simplify notation, $\sigma \in \mathbb{T}$ is used to represent
$\sigma$ belonging to the node set of $\mathbb{T}$. For any node $\sigma \in
\mathbb{T}$ with $\sigma \ne \mathit{root}(\mathbb{T})$, denote
$\mathit{parent}(\sigma) \in \mathbb{T}$ as the parent node of $\sigma$. The Markov
process indexed by $\mathbb{T}$ is a set of random variables 
$\{X_\sigma \in \V: \sigma \in \mathbb{T}\}$ such that 
$X_{\mathit{root}(\mathbb{T})}$ is
initialized from $\pi$ and 
\[
  \Pr(X_\sigma = v | X_{\mathit{parent}(\sigma)} = u) = P(u,v), \mbox{ for } u,v \in \V.
\] 
The distribution of $X_\sigma$ is completely determined by the state of
$X_{\mathit{parent}(\sigma)}$.  \cite{benjamini1994markov} called this process a $(\mathbb{T},P) \textit{-walk
on } \G$. In the social network $\G$, an edge represents friendship. In the
referral tree, a directed edge $(\tau, \sigma)$
represents that random individual $X_{\tau} \in \V$ refers random individual
$X_{\sigma} \in \V$ in the $(\mathbb{T},P)\textit{-walk on } \G$.

\textbf{Statistical estimation.} 
For any function on the nodes of the graph $y:\V \rightarrow \R$, denote
\[
  \mu_{\pi,y} := \rmE_\pi y := \sum_{u \in \V} y(u) \pi(u) \ \mbox{ and } \ 
  \mu_y := \rmE y := \frac 1N \sum_{u \in \V} y(u),
\]
where $N:=|\V|$ is the number of nodes in the social network. By assumption, $X_0
\sim \pi$. So, $X_\tau \sim \pi$ and the sample mean
$1/n \sum_{\tau \in \mathbb{T}} y(X_\tau)$ 
consistently estimates $\mu_{\pi,y}$, the population mean under stationarity.
Thus, it is not a consistent estimator for the parameter of interest, namely
the population mean $\mu_y$. In order to estimate $\mu_y$, one can use inverse
probability weighting (IPW), using the stationary distribution. It can be shown
that  	
\[
  \hat \mu_{IPW} = \frac1n \sum_{\tau \in \mathbb{T}} \frac 1 N \cdot
  \frac{y(X_\tau)}{\pi(X_\tau)} 
\] 
is an unbiased and consistent estimator of $\mu_y$. Typically, $N$ is unknown.
The Hajek estimator circumvents this problem while remaining asymptotically
unbiased,
\begin{equation} \label{eq:sec2:MCest}
  \frac{1}{\sum_{\tau \in \mathbb{T}} 1/\pi(X_\tau)} \sum_{\tau \in \mathbb{T}}
  \frac{y(X_\tau)}{\pi(X_\tau)}.
\end{equation}

The typical ``simple random walk'' assumption in the RDS literature is that
participants select uniformly from their contacts.  This corresponds to $T_{uu}
= \deg(u)$, making $\pi(u) \propto \deg(u)$, which is something that can be asked of
participants.  Under these assumptions, \eqref{eq:sec2:MCest} reduces to
the  RDS II estimator \citep{heckathorn2007extensions}
\[
  \hat{\mu}_y = \frac{1}{\sum_{\tau \in \mathbb{T}} 1/\deg(X_\tau)} 
  \sum_{\tau \in \mathbb{T}} \frac{y(X_\tau)}{\deg(X_\tau)}. 	
\]

\subsection{The Variance of RDS} \label{sec:variance}
Many empirical and social networks display community structures
\citep{girvan2002community}. This can lead to referral bottlenecks in the
Markov chain. These bottlenecks exist because respondents are likely
to refer people within their own community who have similar characteristics. This
section specifies how bottlenecks make successive samples dependent,
increasing the variance of $\hat \mu_{y}$ and the design effect of RDS. The
spectral properties of the Markov transition matrix reveal the strength of
these bottlenecks and control the variance of estimators like $\hat
\mu_{IPW}$. These results motivate the main results of this paper, which show
that anti-cluster sampling improves the relevant spectral properties of the
Markov transition matrix under a certain class of Stochastic Blockmodels. As a
result, anti-cluster sampling can decrease the variance of estimators like
$\hat \mu_{IPW}$.  

Let $\lambda_2(P_A)$ be the second largest eigenvalue of the Markov transition
matrix for the simple random walk. The Cheeger bound demonstrates that the
spectral properties of $P_A$ can measure the strength of these communities.  
See \cite{chung1997spectral} (Chapter 2) and \cite{levin2009markov} (p. 215)
for more details.  This relationship between communities in $\G$ and the
spectral properties of $P_A$ is exploited in the literature on spectral
clustering.  In that literature, $\G$ is observed and the spectral clustering
algorithm uses the leading eigenvectors of $P_A$ to partition $\V$ into
communities \citep{von2007tutorial}.  

Intuitively, if there are strong communities in $\G$ and the node features $y$
are relatively homogeneous within communities, then successive samples $X_i$
and $X_{i+t}$ will likely belong to the same community and have similar values
$y(X_i)$ and $y(X_{i+t})$.  This makes the samples highly dependent; the
auto-covariance $\mathrm{Cov}(y(X_i),y(X_{i+t}))$ will decay slowly as a function of $t$. The next
lemma decomposes the auto-covariance in the eigenbasis of the Markov transition
matrix. This proposition shows that the auto-covariance decays like
$\lambda_2^t$.

The following result applies to any reversible Markov chain with
$|\lambda_2|<1$. In particular, it applies to both $P_A$ (RDS) and $P_W$
(AC-RDS). With a reversible Markov chain, the assumption $|\lambda_2|<1$ is
equivalent to assuming that the chain is irreducible and aperiodic. 

\begin{prop} \label{lem:cov_eq}	
  Let $(X_i)_{i=0}^n$ be a Markov chain with reversible transition
  matrix $P$. Suppose that $X_0$ is initialized with $\pi$, the stationary
  distribution of $P$. For $j = 1, 2, \dots, N$, let $(f_j,\lambda_j)$ be the
  eigenpairs of $P$, ordered so that $|\lambda_i| \ge |\lambda_{i+1}|$.
  Because $P$ is reversible, $f_j$ and $\lambda_j$ are real valued and the
  $f_j$ are orthonormal with respect to the inner product
  $\inner{f_\ell}{f_j}_\pi = \sum_{i \in \V} f_\ell(i) f_j(i) \pi(i)$. If
  $|\lambda_2|<1$, then 

  \[ 
    \mathrm{Cov}( y(X_i), y(X_{i+t}) ) = \sum_{j=2}^{|\V|}
    \inner{y}{f_j}_\pi^2 \lambda_j^t.  
  \]
\end{prop}

In previous research, \cite{bassetti2006examples} and \cite{verdery2015network} used a
similar expression to compute the variance.  

%
%

\subsection{Anti-Cluster Random Walk; Constructing the Weights $W$} \label{sec:acrw}
This subsection describes a Markov model for AC-RDS. Section \ref{sec:theory}
then studies the spectral properties of the resulting AC-RDS Markov transition
matrix. To describe the model we need the following notation.  Let $\cdot$
denote  element-wise matrix multiplication and let $J_{K \times K}$ denote a
$K\times K$ matrix containing all ones. Finally, define the overbar operator
for a  $K\times K$ matrix $B$ as $\bar B:=J_{K \times K}-B$,  so that $\bA
=J_{N \times N}-A$.



This model creates a Markov transition matrix which can be expressed with
matrix notation.  Under the model, if $i$ has one coupon, then the probability
that $i$ refers $j$ is proportional to the $(i,j)^{th}$ element of the matrix
$(A \bA)\cdot A$. To see this, note that the $(i,j)^{th}$ element of $A \bA$ is
the number of nodes $\ell$ that are friends with $i$ but not friends with $j$,
that is
\[
  [A \bA]_{ij}
  = \sum_\ell A_{i\ell} (1 - A_{j\ell}).
\] 
Then, the element-wise multiplication ensures that $i$ is friends with $j$,
yielding the weight matrix $(A \bA)\cdot A$.  

Note that the weight matrix $(A \bA)\cdot A$ is not symmetric and, thus, does
not lead to a reversible Markov chain. However, we can use a second referral
request to augment the first request to ensure reversibility. To this end,
model the referral request ``Please refer someone that knows many people
that you do not know'' as follows:  if $i$ is friends with $j$, then the
probability that $i$ refers $j$ is  proportional to the number of people that
$j$ knows that $i$ does not know. In a similar fashion as above, this request
produces the weight matrix $(\bA A)\cdot A$. 


To implement AC-RDS, choose between $(A \bA)\cdot A$ and 
$(\bA A)\cdot A$ with equal probability by flipping a coin.  Consider the
matrix $\tilde W$ given by
\begin{equation} \label{eq:ac-rdsweights}
  \tilde W = (A \bA + \bA A)\cdot A.
\end{equation}
The $(i,j)^{th}$ element of $\tilde W$ is proportional to the probability that
$i$ refers $j$ in the process described above. By design, $\tilde W$ is
symmetric, making making $P_{\tilde W}$ a reversible Markov transition matrix. 

These ideas for connecting implementation instructions for AC-RDS with the
Markov model are summarized in Table~\ref{table:implementation}. The next
section studies the spectral properties of $P_{\tilde W}$ under a statistical
model for $\G$.
\begin{table}[h]
  \textbf{Implementation instructions compared to the Markov model}
  \centering
  \begin{tabular}{|p{0.19 \textwidth}|p{0.35\textwidth}|p{0.36\textwidth}|} 
	\hline
	\centering Flip a coin &  If heads (type \textbf{A}), & If tails (type \textbf{B}), \\ \hline
  	\centering Implementation Instructions & 
	\cellcolor{gray!10} Ask ``please refer contacts in the target population who don't know many of your contacts.'' & 
	\cellcolor{gray!10} Ask ``please refer contacts in the target population who have many contacts who don't know you.'' \\ \hline
  	\centering Markov model, starting from node $i$ & 
  	List all pairs of nodes $(j,k)$ such that,
  	$(i,j)\in \E$, $(i,k)\in\E$, and $(k,j)\notin \E$. Then choose a pair $(j,k)$
  	 uniformly and refer $j$ or $k$ uniformly at random. &
  	List all pairs of nodes $(j,k)$ such that $(i,j)\in \E$ and $(i,k)\notin \E$.
  	From this list, uniformly choose a node pair $(j,k)$. Refer $j$.\\ \hline
  \end{tabular}
	\caption{The correspondence between AC-RDS implementation instructions and
	the Markov model for the referral process. Referral requests A and B from
	Figure~\ref{fig:refreq} correspond to  the left and right columns,
	respectively, of this table.  The first row describes the verbal request
	given to a participant.  The second row describes the Markov model for this
	request, as discussed in Section~\ref{sec:acrw}.}
  \label{table:implementation}
\end{table}

Finally, we note that the transition matrix $P_{\tilde W}$ does not use
referral request C in Figure \ref{fig:refreq}, ``Please refer someone that does
not know the person that referred you."  Such a request cannot form a Markov
chain on the nodes in the network because it depends on the previous
participant. This non-Markovian behavior should not preclude the use of request
C in practice; however, it does make establishing theoretical results for
request C more difficult. In this paper, we focus on requests A and B and their
Markov transition matrix $P_{\tilde W}$. 

\section{Theoretical Results} \label{sec:theory}
To study the spectral properties of $P_{\tilde W}$ under a statistical model
for the underlying social network, we break the analysis into ``population
results'' and ``sampling results.''  
The ``population results" in this section correspond to using the (weighted)
adjacency matrix $\A = \rmE A$, where the expectation is with respect to the
statistical model for generating the network. The expected adjacency
matrix is a deterministic matrix and various combinatorial techniques can be
used to show its properties.
Define
\begin{equation}\label{eq:pw}
  \tilde \W = (\A \bar{\A} + \bar{\A} \A)\cdot \A.
\end{equation}
Define the Markov transition matrices $P_{\tilde \W}$ and $P_\A$ as in
\eqref{eq:defp}.  In these definitions, $P_\A$ corresponds to the population
matrix for the simple random walk (RDS) and $P_{\tilde \W}$ corresponds to the
population matrix for AC-RDS.

The ``sampling'' referred to in this section introduces an additional layer of
randomness to generate the underlying social network $\G$.  
The goal of ``sample results'' is to show that the random graph generated by
the generic model has similar properties to the expected graph. That is the
randomness of the graph doesn't significantly change the graph from the expected graph.
To refer to the
randomness of the Markov chain, this section will refer to ``anti-cluster
sampling,'' ``Markov sampling,'' or ``respondent-driven sampling.''

The population results will show that under various statistical models for the
underlying social network, the second eigenvalue of $P_{\tilde \W}$ is less
than the second eigenvalue of $P_\A$.   To extend these population results to a
network which is sampled from the model, the sampling results use concentration
of measure to show that $A$ and $\tilde W$ are close (under the operator norm) to $\A$ and $\tilde \W$, respectively.
Then, perturbation theorems show that the eigenvalues of $P_\A$ and $P_{\tilde W}$ are
close to the eigenvalues of $P_A$ and $P_{\tilde \W}$, respectively. Theorem~\ref{thm:var_acrds} 
combines these results with Proposition~\ref{lem:cov_eq} to show that AC-RDS
reduces the covariance between Markov samples. 

\subsection{Population Graph Results}\label{subsec:PopulationGraph}
Anti-cluster sampling is motivated by the need to readily escape communities in
a social network. The Stochastic Blockmodel (SBM) is a standard and popular
model that parameterizes communities in the social network
\citep{holland1983stochastic}. For this reason, the analyses below use the SBM
to study anti-cluster sampling.

\begin{definition}
 To sample a network from the \textbf{Stochastic Blockmodel}, assign each node
 $u \in \{1, 2,\ldots, N\}$ to a class $z(u) \in \{1, 2,\ldots, K\}$,
 where the $z(u)$ are independently generated from Multinomial$(\theta)$.  
 Conditionally on $z$, edges
 are independent and the probability of an edge between nodes $u$ and $v$ is
 $B_{z(u) z(v)}$, for some matrix $B \in [0,1]^{K \times K}$.
\end{definition}

The results below condition on the partition $z$. Conditional on this
partition, $\rmE[A|z]$ has a convenient block structure. Define the partition
matrix $Z \in \{0,1\}^{N \times K}$ such that $Z_{uk} = 1$ if $z(u) = k$,
otherwise $Z_{uk} = 0$.  Define $\A= \rmE[A|z]$ and note that
\[\A = Z B Z^T.\]
Let $\bar{\A} := J_{N\times N} - \A$.  Define the population weighting matrix
as in \eqref{eq:pw}.  The following lemma shows that $\tilde \W$
retains the block structure of $\A$.

\begin{lemma} \label{lem:ac_as_sbm} 
	Define $\bar{B} := J_{K\times K} - B$ and $\Theta  \in \R^{K\times K}$ as
	a diagonal matrix with $\Theta_{kk}$ equal to  the expected number of
	nodes in the $k$th block. Then, $\tilde \W = (\A \bar{\A} + \bar{\A}
	\A)\cdot \A$ can be expressed as
  \begin{align*}
    \tilde  \W = Z\left((B\Theta \bar{B} + \bar{B}\Theta B) \cdot B\right)Z^T.
  \end{align*}
\end{lemma}

The following lemma shows that under a certain class of Stochastic Blockmodels,
anti-cluster sampling decreases the probability of an in-block referral.  
\begin{lemma} \label{lem:ac_increases_leaving_prob} 
  For $0<r<p+r<1$, let $B = pI + rJ_{K\times K}$. 
  If $\Theta_{ll}r < \Theta_{kk}(p+r)$ for all $k \neq l$, then for any two nodes $u$ and
  $v$ with $z(u) = z(v)$,
  \[ P_{\tilde \W}(u,v) < P_\A(u,v). \]   
\end{lemma} 
Note that if every block has an equal population, then the first assumption,
$0<r<p+r<1$, implies the second assumption $\Theta_{ll}r < \Theta_{kk}(p+r)$.
The next proposition uses Lemma~\ref{lem:ac_increases_leaving_prob} to show
that anti-cluster sampling reduces the second eigenvalue of the population
Markov transition matrix. 

\begin{prop}[Spectral gap of the population graph] \label{lem:spgap_popgraph}
  Under the SBM with $K$ blocks, let $B = p I + r J_{K\times K}$, for
  $0<r<p+r<1$.  If the $K$ blocks have equal size, then 
  \begin{equation}\label{eq:balance} 
    0 < \lambda_2(P_{\tilde \W}) + \epsilon < \lambda_2(P_{\A}) < 1,
  \end{equation}
  where $\epsilon >0$ depends on $K,p,$ and $r$, but is independent of $N$, the
  number of nodes in the graph.  Specifically, $\lambda_2(P_{\A}) =
  1/(R+1)$, where $R = Kr/p$. In the asymptotic setting where $K$ grows and $r$
  shrinks, while $p$ and $R$ stay fixed, 
  \begin{equation}\label{eq:smallworld}
    \lambda_2(P_{\tilde \W}) \rightarrow \frac{1}{c R+1}, \ \mbox{ with } \ c = \frac{R+1}{R+1 - p}.
  \end{equation}
\end{prop}
For any single node, note that $R$ is roughly the expected number of
out-of-block edges divided by the expected number of in-block
edges. To see this, multiply the numerator and denominator of $Kr/p$
by the block population $N/K$.  As such, it is approximately the odds that a random
walker will change blocks.  When $R$ is large,  the Markov chain mixes quickly
and $\lambda_2(P_{\A})$ is small to reflect that.  

AC-RDS is most useful in social networks with tight communities, where the walk
is slow to mix; this corresponds to a larger value of $p$ and a smaller value
of $R$.  In this setting, $c$ in \eqref{eq:smallworld} is large, thus
making $\lambda_2(P_{\tilde \W})$ much smaller than $\lambda_2(P_{\A})$.  In
particular, if $p$ is close to one, then $c \approx 1 + R^{-1}$ becomes very
large for small values of $R$.  Notice that the second part of
Proposition~\ref{lem:spgap_popgraph} makes no assumption on $N$, the number
of nodes in the network.


The next proposition shows that anti-cluster sampling continues to perform
well, even when the community structure is exceedingly strong and standard
approaches will fail to mix well.  Here, the reduction of $\lambda_2$ from
anti-cluster sampling is dramatic.

\begin{prop} \label{prop:epsilon}
  Under the SBM with $2$ blocks of equal sizes, let $\epsilon>0$ and suppose
  that $B_{kk}=(1-\epsilon)$ and $B_{kl}=\epsilon$ for $k \neq l$. Then,
  \[\lim_{\epsilon \searrow 0} \lambda_2(P_\A) = 1\]
  and
  \[\lim_{\epsilon \searrow 0} \lambda_2(P_{\tilde \W}) = 1/3.\]
\end{prop}  
For any Markov transition matrix $P$, $\lambda_2(P) \le 1$.   The graph is
disconnected if and only if $\lambda_2=1$; this is the most extreme form of a
bottleneck.  In the above proposition, if $\epsilon = 0$, then the sampled
graph will contain two disconnected cliques, one for each block.  Under this
regime, both $P_A$ and $P_{\tilde W}$ will have second eigenvalues equal to
one. However, if $\epsilon$ converges to zero from above, then 
Proposition~\ref{prop:epsilon} shows that  $\lambda_2(P_{\tilde \W})$
approaches 1/3, while $\lambda_2(P_\A)$ approaches 1.  

Propositions~\ref{lem:spgap_popgraph} and~\ref{prop:epsilon}  assume
balanced block sizes (i.e. an equal number of nodes). To study unbalanced
cases, the necessary algebra quickly becomes uninterpretable.  We explore
the role of unbalanced block sizes with numerical experiments in
Section~\ref{sec:sim}.

\subsection{Sample Graph Results}\label{subsec:SampleGraph} 
Theorem~\ref{thm:concentration_ac_laplacian}  gives conditions which ensure
that the population eigenvalues, $\lambda_\ell(P_{\tilde \W})$, are close to
the sample eigenvalues, $\lambda_\ell(P_{\tilde W})$. As such, the
population results in the previous section appropriately represent the behavior
of Markov sampling (both AC-RDS and RDS) on a network sampled from the
Stochastic Blockmodel. \citet{chung2011spectra} prove a  similar result for
$\left|\lambda_\ell(P_{A}) - \lambda_\ell(P_{\A}) \right|$. 

\begin{theorem}[Concentration of the anti-cluster random walk] 
  \label{thm:concentration_ac_laplacian} Let $\G=(\V,\E)$ be a random graph
  with independent edges and $\A = \rmE A$ be the expected adjacency matrix.
  Let 
  $\D_i:=\sum_k\A_{ik}$,
  $F_{ij}:=\sum_k\A_{ik}(1-\A_{kj})$, and 
  $G_{ij}:=\sum_k(1-\A_{ik})\A_{kj}$. 
  Define $F_{\min}=\min_{i,j=1,\cdots,|\V|} F_{ij}$.
  If 
  $F_{\min} = \omega\left(\ln N\right)$ and
  there exits a constant $c_1$ such that $F_{ij} + G_{ij} \geq c_1 \D_i$
  for all $i,j \in \{1,\cdots,|\V| \}$,
  then with probability at
  least $1-\e$,
  \[
    \left\| T^{-\half} \tilde W T^{-\half} - \T^{-\half}
    \tilde \W \T^{-\half} \right\|^2
    \leq \frac{c_2 \ln\frac {10N}{\e}}{F_{\min}},
  \] 
  where $c_2$ is a constant, 
  $\left\| \cdot \right\|$ denotes the operator
  norm, $T$ is a diagonal matrix with the row sums of $\tilde W$ on its
  diagonal, and $\T$ is defined in the same way with respect to $\tilde \W$.
	Moreover, with probability at least $1-\e$,
  \[
   \left | \lambda_\ell(P_{\tilde W}) - \lambda_\ell(P_{\tilde \W}) \right|^2 
    =O \left( 
    \frac{\ln\frac {10N}{\e}}{F_{\min}}
    \right),
    \ \mbox{ for all } \ell \in 2, \dots,N.
  \]
\end{theorem}

\begin{remark}
The theorem uses standard asymptotic notation, which we recall here for
convenience. We write $f(n) = O \left(g(n) \right)$ to indicate that
$\left| f \right|$ is bounded above by $g$ asymptotically, that is
$$
 \limsup _{n\to \infty }{\frac {\left|f(n)\right|}{g(n)}}<\infty.
 $$
We write $f(n) = \omega \left( g(n) \right)$ to indicate that $f$ dominates $g$ asymptotically, that is
 $$
  \lim _{n\to \infty }\left|{\frac {f(n)}{g(n)}}\right|=\infty. 
 $$
\end{remark}

\begin{remark}
$F_{ij}$ gives the number of friends of node $i$ that are not in the friend
list of node $j$. So $F_{\min} = \omega\left(\ln N\right)$ ensures that the
number of individuals that a node can refer under AC-RDS grows with a rate
faster than $\ln N$. Roughly speaking, it is similar to the sparsity condition
required for  concentration results of random graphs with independent edges.
Since $\A$ is a symmetric matrix, $F_{ij}=G_{ji}$ and, consequently,
\[\min_{i,j=1,\cdots,|\V|} F_{ij} = \min_{i,j=1,\cdots,|\V|} G_{ij}.\]

The condition on $c_1$ ensures that the ratio $\frac{\D_i}{F_{ij}+G_{ij}}$
stays bounded. These sampling results are sufficiently general to apply to all
of the models studied in the previous section. 
\end{remark}

Theorem~\ref{thm:var_acrds} presents the asymptotic behavior of AC-RDS in reducing
the correlation among samples collected from a random graph under a Stochastic
Blockmodel. The theorem is an aggregation of all the previous results in the paper.  
The result is asymptotic in the size of the population, not in the size of the sample.  

\begin{theorem}[Dependency reduction property of AC-RDS] \label{thm:var_acrds}
  Let $\G$ be a random graph with $N$ nodes sampled from a Stochastic Blockmodel with 
  $B = pI_{K\times K} + r J_{K \times K}$, for $0<r<p+r<c<1$. Further assume an equal number of nodes in each of the $K$ blocks.  
  Let $(X_i)_{i=1}^n$ and $(X^{ac}_i)_{i=1}^n$ be two Markov chains with transition matrix $P_A$ and $P_{\tilde W}$, respectively. 

  The parameters $p,r$ and $K$ can change with $N$.  If $\ln(N)/(pK + rN)
  \rightarrow 0$, then asymptotically almost surely, for all $i$,
  $i+t \in \{1, \dots, n\}$, and $t\neq 0$, 
  \[
    \mathrm{Cov}( y(X^{ac}_i), y(X^{ac}_{i+t}) ) < \mathrm{Cov}( y(X_i), y(X_{i+t}) ), 
  \] 
  where  $y:\V\rightarrow \R$ is any bounded node feature.
\end{theorem}
\begin{remark}
The quantity $p\frac NK + rN$ is $\D_{\min}$, the minimum expected degree. The
condition $\ln(N)/(p\frac NK + rN) \rightarrow 0$ is needed to use
Theorem~\ref{thm:concentration_ac_laplacian}. Note that
$F_{ij}+G_{ij}>2cD_{\min}$
for all $i,j \in \{1,\cdots,|\V| \}$. 
\end{remark}

\section{Numerical Experiments} \label{sec:sim} 
We conduct three sets of numerical experiments to compare the performance
of AC-RDS with standard RDS. The first set investigates the impact of
unequal block sizes on the results of Propositions~\ref{lem:spgap_popgraph}
and~\ref{prop:epsilon}. The second set investigates the impact of community
structures and homophily using the Stochastic Blockmodel. In the third set,
we consider an empirical social network with unknown community structure.
Finally, we consider two relaxations of the Markov model to allow for more
realistic settings: sampling without replacement and preferential
recruitment. 

\subsection{The Role of Unequal Block Sizes} \label{subsec:sim_uneq}
In this experiment, we numerically calculate the eigenvalues of $P_{\A}$ and
$P_{\tilde \W}$ under varying SBM parameterizations with $K=2$. Given
$\theta$ and $B$ in the definition of the SBM, we can use results from
\citet{rohe2011spectral} (see the proof of Lemma 3.1) to compute the $K$
non-zero eigenvalues of the transition matrix. 

Consider the setting of Propositions \ref{lem:spgap_popgraph} and
\ref{prop:epsilon} with $K=2$ blocks. These results assume that the blocks
contain an equal number of nodes; here we explore the role of unequal block
sizes. As a measure of unbalance, we use the ratio of the largest block size to
the smallest block size. The results of the study are displayed in Figure
\ref{fig:unbalance}.  The horizontal axis in both panels gives this ratio of
unbalance; when this value is large (farther to the right), the blocks are
exceedingly unbalanced. The vertical axis controls the expected number of
in-block versus out-of-block edges with a parameter $\epsilon$. In the left
panel, $\epsilon$ plays the dual role as in Proposition \ref{prop:epsilon}. In
the right panel, $\epsilon$ does not control the in-block probabilities (i.e.
the diagonal of $B$); here, the diagonal of $B$ is set to $.8$ across all
experiments.

The spectral gap is given by $1-\lambda_2$, we are interested in exploring the ratio
\begin{equation} \label{eq:ratiosg}
  \mbox{ratio of spectral gaps} = \frac{1-\lambda_2( P_{\tilde \W} )}{1 - \lambda_2( P_{\A} )}.
\end{equation}
For a range of unbalances and values of $\epsilon$, Figure~\ref{fig:unbalance}
plots the ratio of spectral gaps. In all of the parameterizations, this value is greater
than one, indicating that anti-cluster sampling decreases $\lambda_2$ relative
to the random walk model of RDS, even with unequal blocks. For example, the contour at 5.3 represents
the class of models such that anti-cluster sampling increases the spectral gap
by over five-fold.
\begin{figure}[h!] 
  \centering
  \textbf{Anti-cluster sampling decreases the sampling dependence.}
  \includegraphics[width=6.5in]{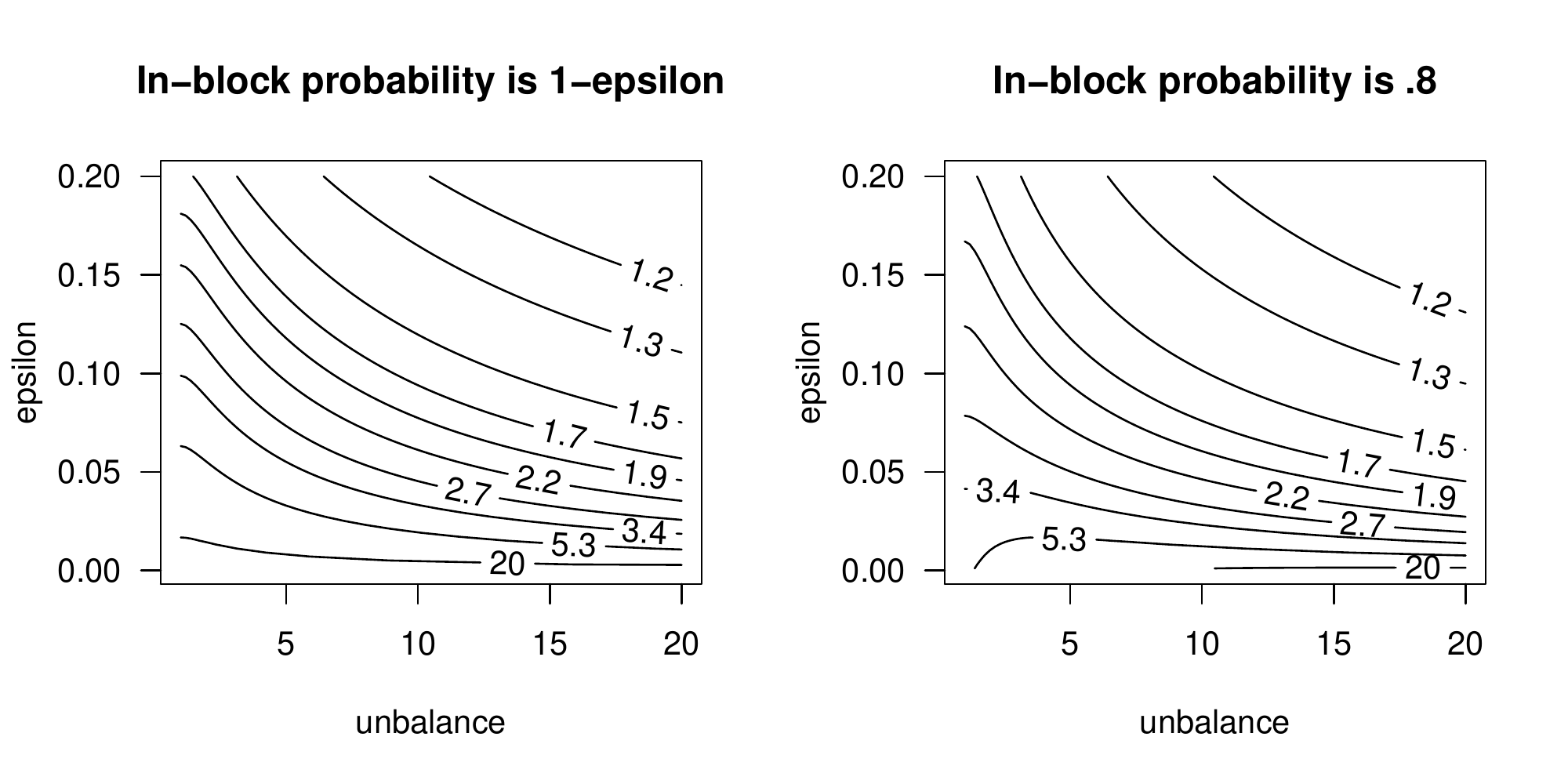} 
	\caption{The ratio of spectral gaps under different parameters of two-block SBM.
	Results for the numerical experiment described in Section
	\ref{subsec:sim_uneq}. This experiment examines the impact of unequal block
	sizes in the setting of Propositions \ref{lem:spgap_popgraph} and
	\ref{prop:epsilon}. As a measure of unbalance (the x-axis), we use the ratio
	of the largest block size to the smallest block size. For a range of SBM
	parameterizations (as described in the text),  these two panels display the
	ratio of spectral gaps as given in \eqref{eq:ratiosg}. All values are greater
	than one, indicating that anti-cluster sampling will increase the spectral
	gap, thus decreasing the dependence between adjacent samples. The benefits of
	anti-cluster sampling are especially prominent when $\epsilon$ is small; this
	corresponds to a model setting in which there are drastically fewer edges
	between blocks.} \label{fig:unbalance} \end{figure}

\subsection{Random Networks} \label{subsec:sim_rand_net}
Here we investigate the impact of community structures and homophily using the
Stochastic Blockmodel. We use a SBM with $2000$ nodes and $50$ communities of
equal size to generate the underlying social network.   To illustrate the
impact of community structures, we vary the ratio of the expected number of in-block edges divided by the expected number of 
out-of-block edges.  This ratio also controls the probability of generating an out-of-community
referral. For example, with the ratio equal to one, the probability of an
out-of-community referral is $1/2$. We examine values of this ratio between $1/2$ to $4$. 
To do this, we fix  the in-block probabilities to $0.9$ and change the out-of-block probabilities.

We simulate Markovian referral trees in which each participant refers exactly
three members with replacement. The three referrals are samples from the
neighbors of the participant. RDS uses uniform samples, whereas AC-RDS uses
non-uniform samples based on the weights described in \eqref{eq:ac-rdsweights}.
To show the effect of the communities, we choose the binary node feature to
be based on the community membership. The value is set to zero if the node
belongs to communities $1$ through $25$, otherwise, the value is set to one. 
For both designs, we use the RDS II estimator to estimate the community
proportion, where the inclusion probabilities are the stationary distribution
of the simple random walk.

The datasets are simulated in the following way. First we generate a
realization of an SBM and compute the stationary distribution of the simple
random walk. We simulate the referral procedure of RDS and AC-RDS starting from
a uniformly selected node and continuing until a certain number of samples are
collected, either $1\%$, $5\%$, or $10\%$ of the total nodes. We compute the
RDS II estimates of the feature from samples collected by both procedures.

This study is based on $5000$ simulated datasets. Figure \ref{fig:rdssbmtree}
displays box plots for the $5000$ RDS II estimates of the proportion in
different settings. Comparing RDS to AC-RDS, we see  that AC-RDS collects more
representative samples. Additionally, as we increase the degree of homophily,
the performance of AC-RDS suffers less. In (a) and (b), the chance that
participants make referrals outside of their community is relatively high,
$2/3$ and $1/2$, respectively. In these cases, both designs perform similarly.
However, in (c) and (d), where there is a smaller chance of cross-community
referral, there is a stronger referral bottleneck.  In this regime, AC-RDS
collects more representative samples by encouraging
participants to leave their communities more often. This is exactly the
intended outcome of AC-RDS. In fact, at the population level, this is the
result proven in Lemma \ref{lem:ac_increases_leaving_prob}.

\begin{figure}[H] 
  \centering\textbf{Estimates with samples from the SBM collected under AC-RDS and RDS.}
  \includegraphics[scale=0.6]{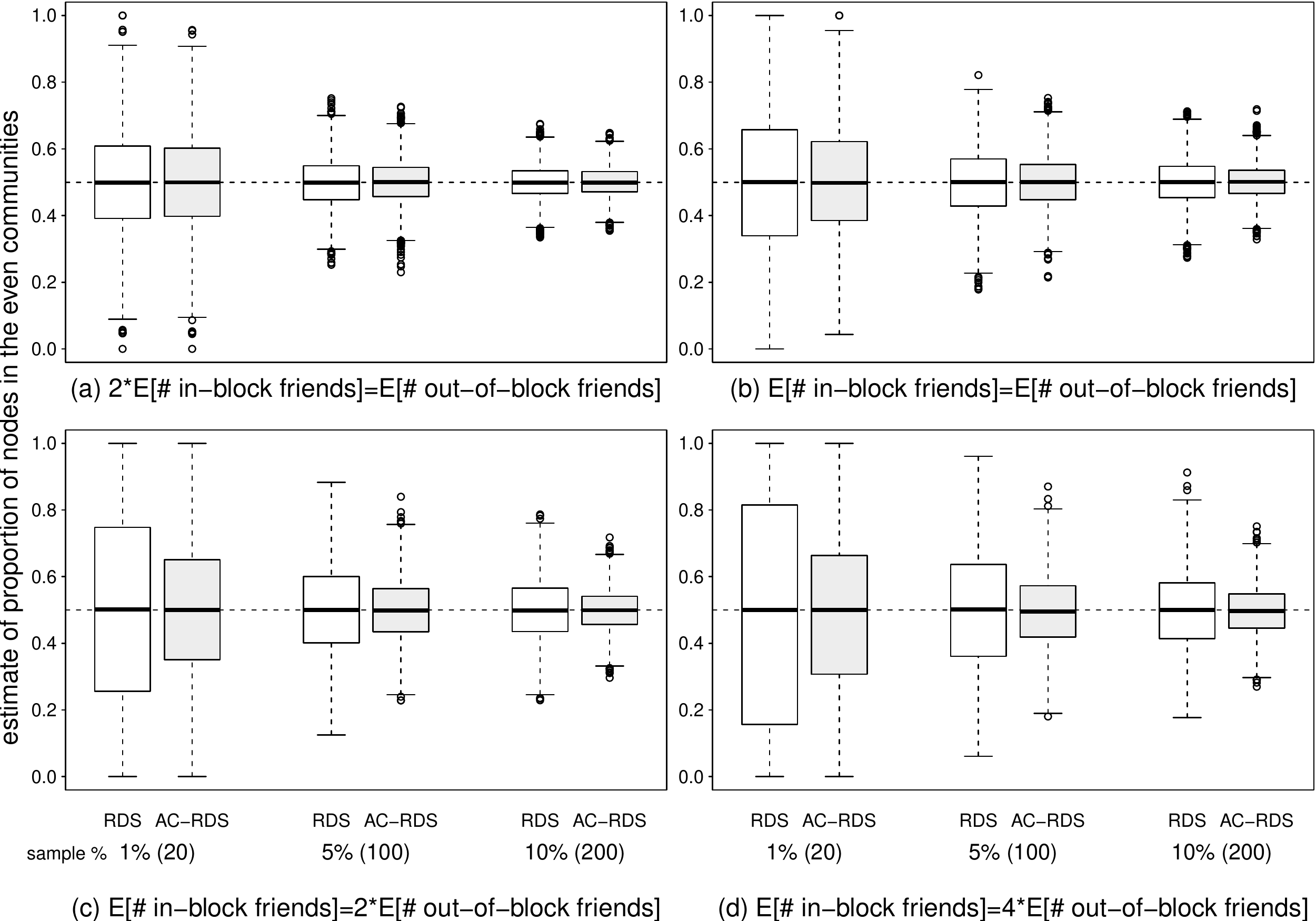} 
	\caption{Simulation results for the random network study described in Section
	\ref{subsec:sim_rand_net}.  The box plots display the estimated proportions
	across the $5000$ simulated datasets. The four panels correspond to four
	different strengths of referral bottlenecks.  The referral bottleneck is the
	strongest in the lower right panel.  Within a panel, there are three pairs of
	box plots, corresponding to three different sample sizes. In this setting,
	AC-RDS dramatically reduces the interquartile range of the estimator. 
}
\label{fig:rdssbmtree}
\end{figure}


\subsection{Add-Health Networks} \label{subsec:sim_addhealth_net} 
This set of simulations is based on friendship networks from the National
Longitudinal Survey of Adolescent(available at
http://www.cpc.unc.edu/addhealth), which we refer to as the Add-Health Study.
In the study, the students were asked to list up to five friends of each
gender, and whether they had any interaction within a certain period of time.
The reported friendships were then combined into an undirected network. That
is, an edge connecting two students means that either student, not necessarily
both, reported a friendship.  We use the four largest networks in the dataset.
Table \ref{table:addhealthinfo} contains summary information for the largest
connected component of these four networks. We use gender as the binary node
feature and focus on estimating the proportion of males in the population.
\begin{table}[h]
  \centering
  \begin{tabular}{l|l l l l l}
    School id & \# Nodes &\# Edges & CC    & covariance & covariance$^{ac}$\\ \hline
    School 36 & 2152     &  7986   & 0.178 & 0.0260 & 0.0056    \\ \hline
    School 40 & 1996     &  8522   & 0.144 & 0.0265 & 0.0030    \\ \hline
    School 41 & 2064     &  8646   & 0.139 & 0.0243 & 0.0042    \\ \hline
    School 50 & 2539     & 10455   & 0.141 & 0.0276 & 0.0069    \\
  \end{tabular}
  \caption{Network characteristics for the four largest friendship networks in
    the Add-Health study. This table provides characteristics for the largest
    connected component of each network. An edge between student nodes indicates
    that either student reported a friendship. The clustering coefficient (CC) is
    the ratio of the number of triangles and connected triplets. The last two
    columns represent the covariance of the samples collected under RDS and AC-RDS,
  respectively.}
\label{table:addhealthinfo}
\end{table}

We simulate the referral procedure of RDS and AC-RDS starting from a uniformly
selected node and continuing until a certain number of samples are collected,
either $1\%$, $5\%$, or $10\%$ of the total nodes. In these simulations, each
participant refers exactly three members with replacement.  We compute the RDS
II estimate of the male proportion using the node degree for the weights.  
Similar to the simulations in \cite{goel2010assessing} and \cite{baraff2016estimating}, 
these simulations are performed with replacement.  

This study is based on $10,000$ simulated samples. Figure
\ref{fig:rdsaddhealthtree} and \ref{fig:rdsaddhealthtreeTwoUnifW} display box
plots for the $10,000$ RDS II estimates of the male proportion under different
settings. Notice that in Figure \ref{fig:rdsaddhealthtree} the interquartile
range of AC-RDS with a 5\% sample is often comparable to the interquartile
range of a standard RDS with a sample that is twice as large. In Figure \ref{fig:rdsaddhealthtreeTwoUnifW},
only type B request is considered in the implementation of AC-RDS.

\begin{figure}[H] 
	\textbf{Estimates with samples from the Add-Health friendship networks
		collected under AC-RDS and  RDS.}
  \centering
  \includegraphics[scale=0.65]{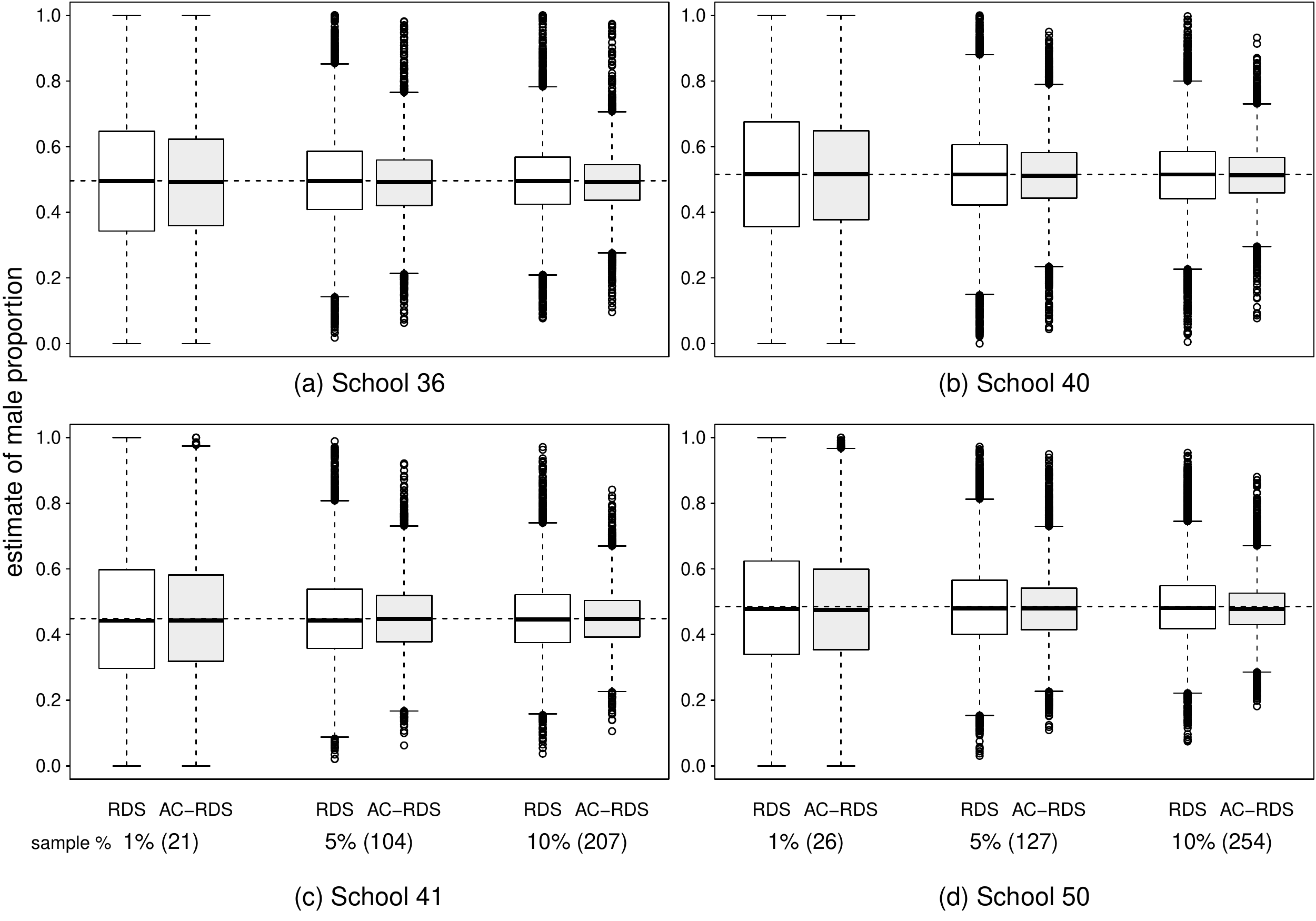} 
	\caption{Simulation results based on the Add-Health study described in
	Section \ref{subsec:sim_addhealth_net}.  The box plots display the estimated
	proportion of men across the $10,000$ simulated samples.  Each panel
	corresponds to a different network from the study.  Within a panel, there are
	three pairs of box plots, corresponding to three different sample sizes.  The
	results compare the RDS II estimator based upon (1) a standard RDS sample and
	(2) an AC-RDS sample. Notice that the interquartile range of AC-RDS with a
	5\% sample is often comparable to the interquartile range of a standard RDS
	with a sample that is twice as large.}
  \label{fig:rdsaddhealthtree}
\end{figure}

\begin{figure}[H] 
	\textbf{Estimates with samples from the Add-Health friendship networks
	collected under AC-RDS type \textbf{B} and  RDS.}
  \centering
  \includegraphics[scale=0.65]{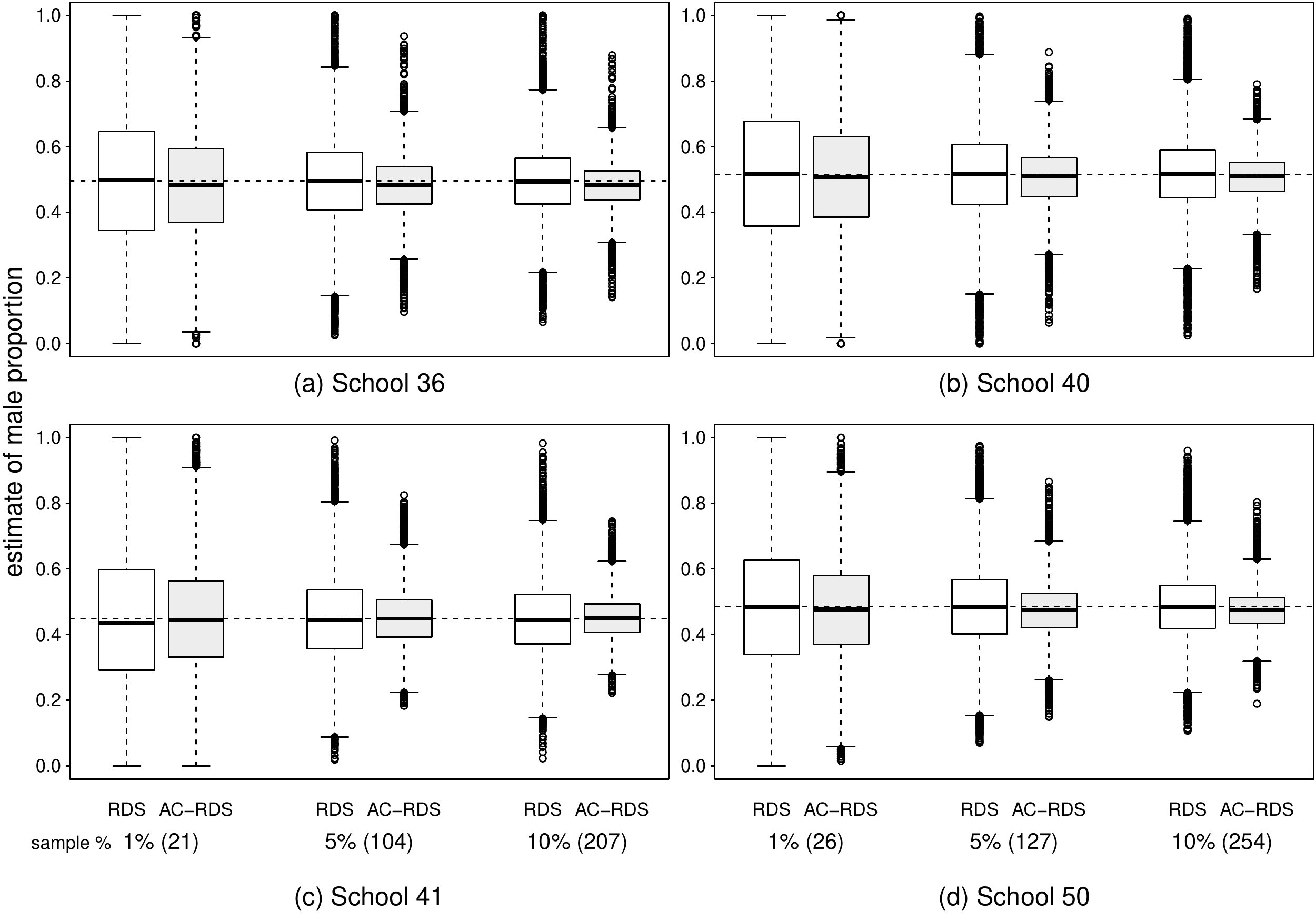} 
	\caption{Simulation results based on the Add-Health study described in
	Section \ref{subsec:sim_addhealth_net}.  The box plots display the estimated
	proportion of men across the $10,000$ simulated samples.  Each panel
	corresponds to a different network from the study.  Within a panel, there are
	three pairs of box plots, corresponding to three different sample sizes.  The
	results compare the RDS II estimator based upon (1) a standard RDS sample and
	(2) an AC-RDS type \textbf{B} sample.}
  \label{fig:rdsaddhealthtreeTwoUnifW}
\end{figure} 

\newpage 

\subsection{Without Replacement Sampling} \label{subsec:sim_wo_replacement}
We consider the impact on AC-RDS when simulating the sample with and without replacement 
from the underlying network.  In the Random Networks simulation
model, there is only a small difference between the two sampling
settings.  This is likely because the network is dense. In smaller networks,
one expects there to be a greater difference between with and without replacement 
sampling.  In fact, in the Add-Health simulation model, under a
without replacement setting and a referral rate of one or two, the trees die
quickly and often do not collect enough samples to attain $1\%$ of the total
nodes. Figure~\ref{fig:sim_unif_seed_WO} displays plots for the $10,000$
RDS II estimates of the male proportion under the without replacement setting.
In the simulation study of Add-Health networks type B implementation of AC-RDS
collects more representative samples compare to the two types combined implementation.

\begin{figure}[H]
	\centering
	\textbf{Estimates with samples from the Add-Health friendship networks
	collected under AC-RDS type B and RDS \textit{without replacement}.}
	\includegraphics[width=1\textwidth]{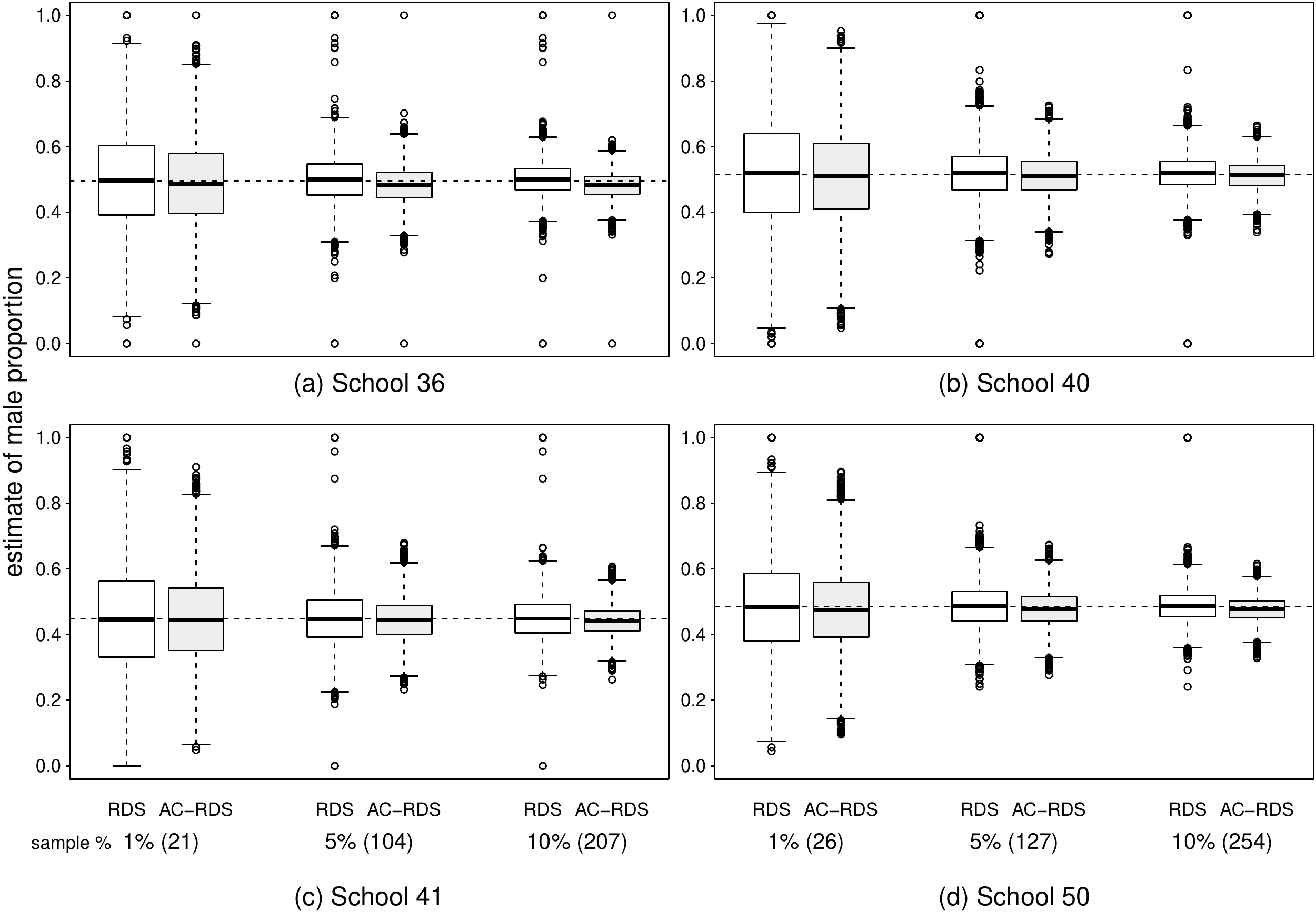}
	\caption{Without replacement simulation results based on the Add-Health
	study described in Section \ref{subsec:sim_addhealth_net}.  The box plots
	display the estimated proportion of men across the $10,000$ simulated
	samples.  Each panel corresponds to a different network from the study.
	Within a panel, there are three pairs of box plots, corresponding to three
	different sample sizes.  The results compare the RDS II estimator based upon
	(1) a standard RDS sample and (2) an AC-RDS type B sample. Notice that the
	interquartile range of AC-RDS with a 5\% sample is often comparable to the
	interquartile range of a standard RDS with a sample that is twice as large. } 
	\label{fig:sim_unif_seed_WO}
\end{figure}

\subsection{Non-uniform seeds} \label{subsec:sim_nonuniform_seed}
We consider the impact of non-uniform (biased) seed nodes on AC-RDS 
and standard RDS when simulating the sample with and without replacement from the underlying network.
Figure~\ref{fig:sim_nonunif_seed} displays plots for the $10,000$ RDS II
estimates of the male proportion under the biased seed nodes.

\begin{figure}[H]
	\centering
	\textbf{\small{Estimates with samples from the Add-Health friendship networks
	collected under AC-RDS type \textbf{B} and  RDS with non-uniform seed nodes.}}
	\begin{tabular}{m{0.01\textwidth} m{1\textwidth}}
		\textbf{A} & \includegraphics[width=0.85\textwidth]{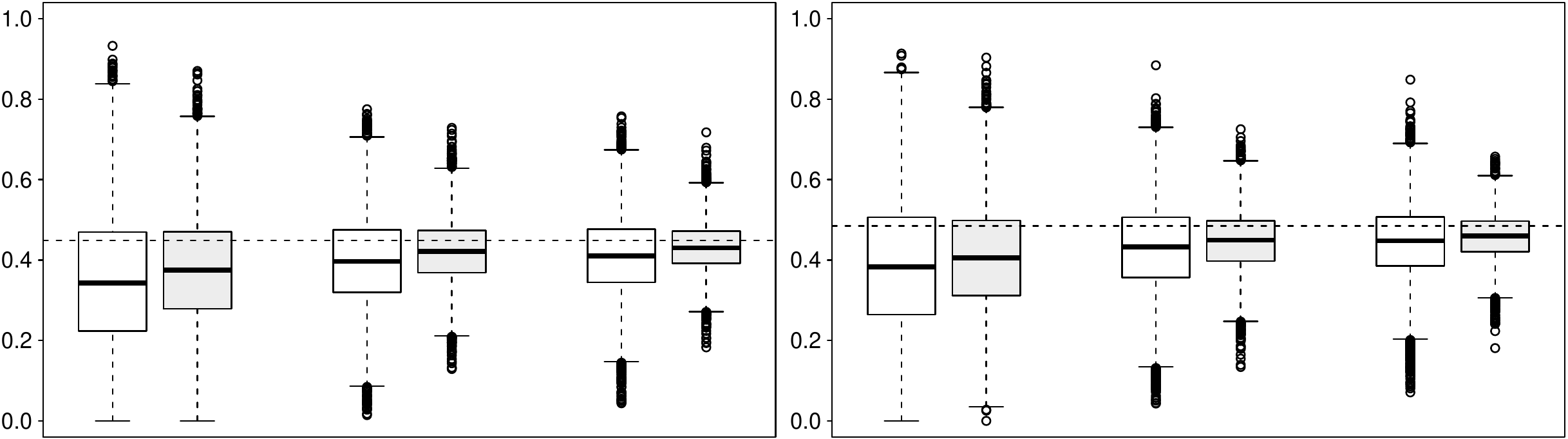}  \\
		\textbf{B} & \includegraphics[width=0.85\textwidth]{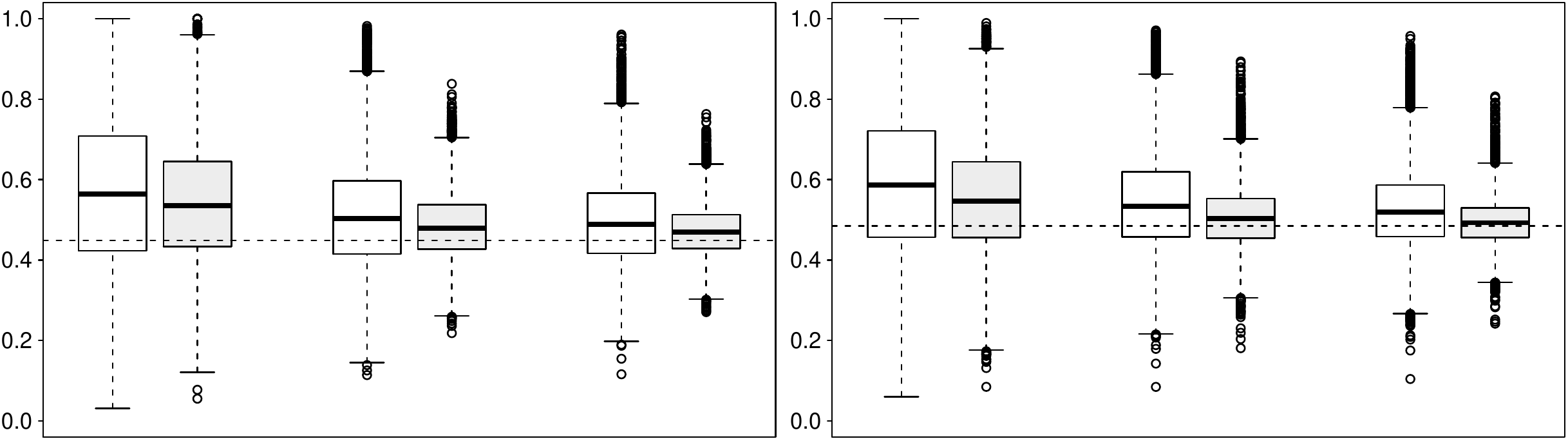}  \\
		\textbf{C} & \includegraphics[width=0.85\textwidth]{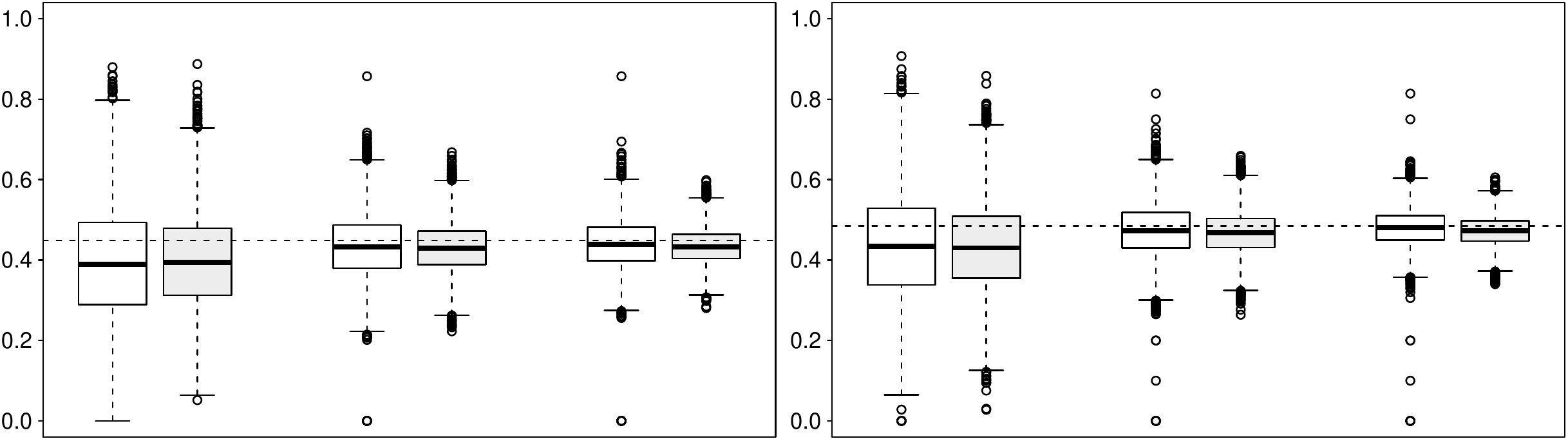} \\
		\textbf{D} & \includegraphics[width=0.85\textwidth]{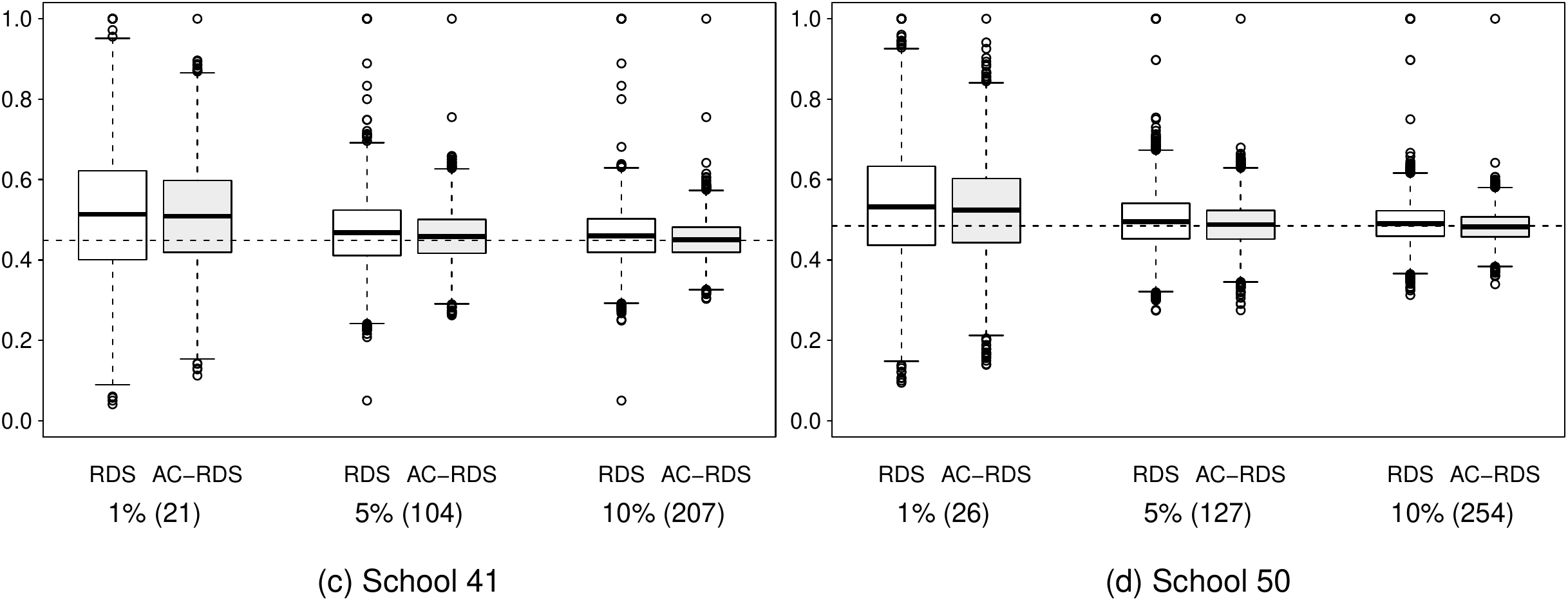}
	\end{tabular}
	\caption{\small{Simulation results based on the Add-Health study described in
	Figure~\ref{fig:rdsaddhealthtreeTwoUnifW}. Each column corresponds to a
	different network from the study. In Row A and B the samples are collected
	with replacement and in Row C and D without replacement. In Row A and C the
	seed node is chosen from node with female attributes and in Row B and D from
	nodes with male attributes.}} 
	\label{fig:sim_nonunif_seed}
\end{figure}


\section{Issues remaining} \label{sec:implementation}
The aim of this research is to highlight how referral requests have the
potential to alter referral patterns in a way that makes the resulting sample
more representative of the target population.  The Markov models for RDS and
AC-RDS capture important features of reality, but both are necessarily an
approximation to the practicalities of gathering a sample from a marginalized
and hard-to-reach population.  These gaps between ``theory'' and ``practice''
have the potential to make AC-RDS either more or less desirable. If AC-RDS is
to be implemented in the field, there are several issues that must be explored. 
\begin{enumerate}
\item If there are pockets of the marginalized target population which are
	particularly hard-to-reach, AC-RDS has the potential to both help and hinder
	the sampling of these populations.  Novel referral requests could help by
	encouraging participants to refer friends from different communities,
	potentially exposing a new community to the researchers.  Alternatively,
	because AC-RDS referral requests are likely more difficult for participants,
	it could reduce the number of referrals that are made, making it more
	difficult to reach a target sample size.

\item Because AC-RDS leads to a reversible Markov chain, there exist
	formulations for the sampling weights, akin to the Volz-Heckathorn estimator
	\citep{volz2008probability}. 
	Successive Sampling model.  The formulation of the sampling weights for
	AC-RDS could follow a similar argument as the Volz-Heckathorn weights.
	Because the Volz-Heckathorn estimator assumes a reversible Markov transition
	matrix $P_A$, the stationary distribution is proportional to the row sums of
	$A$ (i.e. the node degrees).  Since AC-RDS also assumes a reversible Markov
	transition matrix $P_{\tilde W}$, the stationary distribution of AC-RDS is
	proportional to the row sums of $\tilde W$ (i.e. $\pi$ given in
	\eqref{eq:stationary}).  In both cases ($A$ and $\tilde W$), the weights
	require asking participants questions about their local social network. 
	Recently, \cite{verdery2016new} introduced data collection methods,
	survey questions, and estimators for RDS to estimate clustering properties
	of the underlying social network. Their estimators are designed to count the number of
	connected triplets and triangles which a participant belongs to. The collected data is
	the main part of estimating the sampling weights in AC-RDS.

\item Preferential recruitment, the tendency of participants to refer
	particular friends, leads to the violation of uniform referral assumption. AC-RDS
	gives participants some instructions for the new referrals. These
	instructions, since they are more specific, may lead to a referral process
	that satisfies the initial assumptions more.  However, studying the reactions
	of members of a hidden population to this type of requests and the impact of
	preferential recruitment on AC-RDS requires rigorous field study that we
	will address in future research.

\item More generally, it is necessary to investigate how human subjects
	consider both standard and non-standard referral requests.  Because it is
	practically infeasible to use random number generators to ensure participants
	refer \textit{randomly chosen} friends, all statistical approaches to RDS
	\textit{assume} that participants refer a random collection of friends.
	Whether these statistical models lead to adequate approximations of the
	actual referral process is an empirical question that has received some
	attention and deserves more. 
	In practice, there are many conditions that are often appended to this
	request. These conditions help define your contacts (e.g. as people you
	(i) know on a first name basis, (ii) have seen in the last month, and (iii)
	fit the eligibility criteria for the study).  Page 330 of \cite{RDSmanual}
	and Appendix Q in \cite{CDCmanual} give further discussion on this topic.  Of
	particular interest is that \cite{RDSmanual}, in the section titled ``Script
	for explaining the recruitment process,’’ says 
	\begin{quote} If possible, try and give the coupons to different types of
		people who you know. (e.g.  different ages, different levels of income,
		from different locations in this city).
	\end{quote} The AC-RDS requests provide a formalization for this exact concept.
	For example, \cite{wejnert2008web} 
	designed a web-based method to sample undergraduate students and study the
	effectiveness and efficacy of RDS; \cite{mccreesh2012evaluation} compared
	an RDS sample in Uganda with a total population survey on the same
	population;  \cite{mccoy2013improving} studied how manipulating incentives
	might change referral patterns; \cite{gile2014diagnostics} proposed
	statistical diagnostics to examine the convergence properties; and
	\cite{arayasirikul2015qualitative} performed qualitative follow-up interviews
	to ask participants about difficulties in finding referrals. Similar
	techniques could be used to evaluate whether novel referral requests provide
	a more representative sample. 
\end{enumerate}

\section{Discussion} \label{sec:final} 
In respondent-driven sampling, bottlenecks create dependencies between the samples;
successive samples are more likely to belong to the same community.  Because of
these dependencies, bottlenecks increase the variability
of the resulting estimators.  While researchers cannot alter the social network to
diminish bottlenecks, researchers can use novel implementations of RDS to implicitly
encourage participants to refer friends in different communities.  In
comparison to other such techniques in the literature, AC-RDS does not require
participants to reveal sensitive information, nor does it require \textit{a
priori} knowledge on what forms the bottlenecks (e.g.  race, gender,
neighborhood, some combination of these factors, or some entirely different
factors). In a closer look, AC-RDS, similar to the ``Search'' mode of NSM, increases
the referrals that are more likely to lead to the unexplored parts of the
network. NSM aims to efficiently explore the network by targeting the best nodes for each
sampling wave, while AC-RDS tries to find and explore the best local edges. Direct
comparison of these two methods requires human subject experiments and is
beyond the scope of the current paper.

We call this approach anti-cluster RDS. This terminology stems from two
distinct, but related, definitions of ``clustering'' in networks. First, the
classical use of ``clustering'' in social networks is the clustering
coefficient, a summary statistic of a network which describes the propensity of
nodes to form triangles. This idea of ``clustering'' is a local measure. The
second form of ``clustering'' is more global and is often used synonymously
with community structure; the idea is that ``clusters'' of individuals form
communities. Both of these types of clusters emerge due to homophily, the
tendency of individuals to become friends with people who are similar. As such,
homophily produces a local-global duality in ``clustering.'' AC-RDS requests
are built upon local structures in the network (which of your friends are
friends) and immediately access the global network patterns, which could be
unknown to the researchers and/or participants.

This paper shows that AC-RDS is analytically tractable under the Markov model.
One key benefit of the specific construction is that $X_i^{ac}$ is reversible
and its Markov transition matrix can be expressed with the underlying adjacency
matrix and standard matrix operations \eqref{eq:ac-rdsweights}.  A key limitation of the
Markov model is that it samples with replacement, while in practice the
sampling is done without replacement.  For further discussion of this topic,
see \cite{ott2016unequal}. The simulations in Section \ref{sec:sim} show that
the key insights from the Markov model continue to hold under the sampling
without replacement model, so long as the sample size is not comparable to the
population size.

Section \ref{sec:theory} studies theoretical properties of AC-RDS. We first
argue that AC-RDS can be approximated by a reversible Markovian process.
Propositions \ref{lem:spgap_popgraph} and \ref{prop:epsilon} show that AC-RDS
can decrease $\lambda_2$, the second eigenvalue of the Markov
transition matrix, on the population graph. Theorem
\ref{thm:concentration_ac_laplacian} shows that these gains from Propositions
\ref{lem:spgap_popgraph} and \ref{prop:epsilon} will continue to hold if the
graph is sampled with independent edges. In addition, Theorem
\ref{thm:var_acrds} shows that AC-RDS reduces the covariance of the samples in
the referral tree under the Stochastic Blockmodel with equal block sizes.

Finally, in Section \ref{sec:implementation} we discuss some of the gaps
between theory and practice, acknowledging that more work needs to be done
before AC-RDS could be implemented in the field. For example, it is not clear
how participants will actually respond to AC-RDS requests. Addressing this
issue requires human subject experiments that are beyond the scope of the
current paper. We are addressing this problem in concurrent research.

\bibliographystyle{natbib}
\bibliography{designedRDS2}

\begin{thebibliography}{}
\bibitem[White {\em et~al.}(2012)White, Lansky, Goel, Wilson, Hladik, Hakim,
  and Frost]{white2012respondent}
White, R.~G., Lansky, A., Goel, S., Wilson, D., Hladik, W., Hakim, A., and
  Frost, S.~D. (2012).
\newblock Respondent driven sampling: where we are and where should we be
  going?
\newblock {\em Sexually Transmitted Infections\/}, {\bf 88}(6), 397--399.
\end{thebibliography}

\newpage
\section*{Appendix}
\appendix 
This appendix provides the proofs contained in the main document. We begin by
presenting some preliminary lemmas. We then provide the proofs for the results
given in Sections~\ref{sec:variance},~\ref{subsec:PopulationGraph}, 
and \ref{subsec:SampleGraph}.

\section{Preliminary Lemmas} 
This section contains lemmas which are used to prove our main results.
Lemmas~\ref{lem:ac_as_sbm} and~\ref{lem:ac_increases_leaving_prob} are
contained in the main paper; we start the preliminary results with 
Lemma~\ref{lem:diag_norm_invariant}. First we state two standard results, given
here for convenience.
\begin{lemma} \label{lem:diag_norm_invariant} 
  Let $A$ be a symmetric matrix and $D$ a diagonal matrix. Then 
  \[\|DA\|  = \|D^{\half}AD^{\half}\|.  \]
\end{lemma}

\begin{lemma}[Bernstein's inequality] \label{lem:bernineq} 
  Let $X_1,\cdots ,X_N$ be independent random variables and $|X_i - \rmE X_i| \leq
  S$ for $i=1,\cdots, N$. Let $\sigma^2 := \sum_{i=1}^N \rmE[X_i - \rmE X_i]^2$. Then
  for all $t \geq 0$,
  \[
    \Pr\left(\left| \sum_{i=1}^N X_i - \rmE X_i \right| \geq t \right) 
    \leq 2 \exp\left(-\frac{\frac 12 t^2}{\sigma^2 + \frac 13 St}\right).
  \]
\end{lemma}

We use the following result from \cite{rohe2011spectral} 
in the proof of Proposition \ref{lem:spgap_popgraph}.
\begin{lemma}\label{fact:lam2}[\cite{rohe2011spectral}]
Under the Stochastic Blockmodel, if $B = pI + rJ$ and there are an equal number
of nodes in each block, then 
\[
\lambda_i(P_{\A}) = \left\{\begin{array}{cl}1 & i = 1 
  \\(Kr/p + 1)^{-1} & i = 2, \dots, K \\0 & o.w.\end{array}\right.
\]
\end{lemma}
For completeness we include the proof here. 
\begin{proof}

  The matrix $B \in \R^{k \times k}$ is the sum of two matrices,
  \[B =  pI + rJ_{k}\1_{k}^T,\]
  where $I_{k} \in \R^{k \times k}$ is the identity matrix, $\1_{k} \in \R^{k}$
  is a vector of ones, $r\in (0,1)$ and $p \in (0,1-r)$. 
  Let $Z \in \{0,1\}^{N \times K}$ be such that $Z^T \1_N = s \1_K$ for some $s
  \in \R$. This guarantees that all $K$ blocks have equal size $s$. The
  Stochastic Blockmodel has the population adjacency matrix, $\A =Z BZ^T$.
  Moreover, $P_{\A} = Z B_L Z^T$ for 

  \[B_L =  \frac{1}{Nr + sp} \left(pI_{K} + r\1_{K}\1_{K}^T\right).\]

  The eigenvalues are found by construction.
  \begin{itemize}
    \item The constant vector $\1_N$ is an eigenvector with eigenvalue $1$;
    \begin{align*}
	    Z B_L Z^T \1_N 
        &=\frac{s}{Nr + sp} Z  \left(pI_{K} + r\1_{K}\1_{K}^T\right) 1_K \\
        &=\frac{s}{Nr + sp} Z  (p + K r)\1_{K} 
        +\frac{s(p + K r)}{Nr + sp} \1_{N} = \1_N,
      \end{align*}
      where the last line follows because $N = sK$.
    \item Let $b_2, \dots, b_K \in \R^K$ be a set of orthogonal vectors which
      are also orthogonal to $\1_K$.  For any $i$, $Z b_i$ is an eigenvector
      with eigenvalue $(Kr/p + 1)^{-1}$,
      \[Z B_L Z^T (Z b_i)  = 
        Z B_L s I_{K \times K} b_i  = 
        \frac{s}{Nr + sp} Z  \left(pI_{K} + r\1_{K}\1_{K}^T\right) b_i  = 
      \frac{p s}{Nr + sp} (Z   b_i).\]
      Because $Zb_i$ and $Zb_j$ are orthogonal for $i \ne j$, the multiplicity
      of the eigenvalue $(Kr/p + 1)^{-1}$ is at least $K-1$. 
  \end{itemize}
  Because $rank(P_{\A}) \le min( rank(Z) , rank(B_L) ,rank( Z^T))\le K$, there
  are at most $K$ nonzero eigenvalues.  The results follow.
\end{proof}

The following result is used for the computation of the eigenvalues in the proof of Proposition \ref{prop:epsilon}. 
\begin{lemma}\label{lem:lam2}
  Let $P$ be a block constant Markov transition matrix, with blocks of identical sizes.  Let $P$ contain the block values
  \[P = \left(\begin{array}{c|c}p & r \\\hline r & p\end{array}\right),\]
  then 
  \[\lambda_2(P) = \frac{p-r}{p+r}.\]
\end{lemma}
\begin{proof}
This follows from Lemma \ref{fact:lam2} using $K=2$.
\end{proof}

\begin{lemma}[Operator norm of non-negative irreducible matrices] 
  \label{lem:posspup}
  Let $A \in R^{N\times N}$ be a non-negative, irreducible matrix. Define
  $r_i(A):=\sum_{j=1}^N A_{ij}$. Then 
  \[\|A\| \leq \max_i r_i(A).\] 
    \begin{proof}
By the Perron-Frobenius theorem, $A$ has a real leading eigenvalue.  Additionally, for any $y \in \R^N$, $\mu \in \R$,  with $y \geq 0$, and $\mu \geq 0$, if $Ay \leq \mu y$,  then $\lambda_1(A) \leq \mu$.  Take $y = \1$ and $\mu = \max_i r_i(A)$, then
\[ \|A\|  = \lambda_1(A) \leq \max_i r_i(A). \]
  \end{proof}
\end{lemma}

\begin{lemma}\label{lem:symrw}
  For any $W \in R^{N \times N}$, define diagonal matrix $T$ to contain the row
  sums down the diagonal, $T_{uu} = \sum_{v} W(u,v)$.  If $T_{uu} > 0$ for all
  $u$, then the eigenvalues of $P_W = T^{-1}W$ are equal to the eigenvalues of
  $L_W = T^{-1/2}WT^{-1/2}$.  
\end{lemma}
\begin{proof}
  Let $x,\lambda$ be an eigenpair of $L_W$,
  \begin{equation*}
    T^{-1/2}WT^{-1/2} x = \lambda x \implies
    T^{-1/2}\left(T^{-1/2}W\left(T^{-1/2} x\right)\right) = \lambda
    \left(T^{-1/2}x\right),
  \end{equation*}
  where the left hand side is $P_W (T^{-1/2}x)$.
  This implies that $T^{-1/2}x, \lambda$ is an eigenpair of $P_W$.  
\end{proof}


\section{Design Effect and Variance} 
Here we provide the proof of Proposition \ref{lem:cov_eq} from Section \ref{sec:variance}.
\begin{proof}[Proof of Proposition \ref{lem:cov_eq}] \label{proof:lem:cov_eq}
  Lemma 12.2 in \cite{levin2009markov} shows that (i) $f_j$ and $\lambda_j$ are
  real valued and (ii) the $f_j$ are orthonormal with respect to
  $\inner{f_\ell}{f_j}_\pi$.  Because $\lambda_2<1$, $f_1$ is the constant
  vector.  We can express the covariance as
  \begin{align}\label{equ:Prop1:a}
    \mathrm{Cov}\left(y(X_i),y(X_{i+t})\right) 
    &= \rmE\left[ ( y(X_i)-\rmE[y(X_i)] )
        (y(X_{i+t})-\rmE[y(X_{i+t})] )\right] \notag \\
    &= \rmE\left[ y(X_i) y(X_{i+t}) \right] - \rmE^2[y(X_1)] \notag \\
    &= \rmE\left[ y(X_1) y(X_{1+t}) \right] - \rmE^2[y(X_1)].
  \end{align}
  Consider the first term of \eqref{equ:Prop1:a}
  \begin{align*}
    \rmE [y(X_1)y(X_{1+t})]
    &= \sum_{u,v\in \V} y(u)y(v) \Pr(X_1=u, X_{1+t}=v) \\
    &= \sum_{u,v\in \V} y(u)y(v) \pi_u P^t(u,v)\\
    &= \sum_{u,v\in \V} y(u)y(v) \pi_u \pi_v 
         \sum_{j=1}^{|\V|} f_j(u)f_j(v) \lambda^t_j \\
    &= \sum_{u,v\in \V} y(u)y(v) \pi_u \pi_v 
      \{ 1+\sum_{j=2}^{|\V|} f_j(u)f_j(v)\lambda^t_j \} \\
    &= \sum_{u,v\in \V} y(u)y(v) \pi_u \pi_v 
     + \sum_{j=2}^{|\V|} \lambda^t_j \sum_{u,v\in \V} y(u)y(v) \pi_u \pi_v f_j(u) f_j(v) \\
    &= \rmE^2[y(X_1)] + \sum_{j=2}^{|\V|} \inner{y}{f_j}_\pi^2 \lambda^t_j.\\
  \end{align*} 
  Hence,
  \[ 
    \mathrm{Cov}\left(y(X_i),y(X_{i+t})\right) 
    =\sum_{j=2}^{|\V|} \inner{y}{f_j}_\pi^2 \lambda^t_j. 
  \]

\end{proof} 


\section{Population Graph Results} \label{app:population_results}
Here we provide the proofs of the results given in Section \ref{subsec:PopulationGraph}---Lemmas~\ref{lem:ac_as_sbm},~\ref{lem:ac_increases_leaving_prob} and  Propositions~\ref{lem:spgap_popgraph},~\ref{prop:epsilon}.

\begin{proof}[Proof of Lemma \ref{lem:ac_as_sbm}] 
  From the definition of $Z$ and $\bAc$ it follows that $Z^TZ = \Theta$ and $\bAc=J_{n\times n} - ZBZ^T=Z\bar{B}Z^T$. Then, 
  \[\A \bar{\A } = ZBZ^TZ\bar{B}Z^T=ZB\Theta \bar{B}Z^T\] 
  and similarly, 
  \[ \bar{\A}\A=Z\bar{B}\Theta BZ^T. \] 
  Hence, 
  \begin{align*}
    (\A \bar{\A} + \bar{\A}\A) \cdot \A &= \left( Z(B\Theta \bar{B} + \bar{B}\Theta B)Z^T \right) \cdot (ZBZ^T)  \\
    &=Z \left( (B\Theta \bar{B} + \bar{B}\Theta B)\cdot B \right) Z^T. 
  \end{align*}
\end{proof}

\begin{proof}[Proof of Lemma  \ref{lem:ac_increases_leaving_prob}] 
We first show that 
  \begin{align} \label{lem:diag:ineq1}
    \frac{[(B\Theta \bar{B})\cdot B]_{kl}}{[(B\Theta \bar{B})\cdot B]_{ll}} > \frac{B_{kl}}{B_{kk}} = \frac{r}{p+r}.
  \end{align}
  We have
  \begin{align*}
    [(B \Theta \bar{B}) \cdot B]_{kl} &= r(\Theta_{kk} (p+r)(1-r) + \Theta_{ll} r(1-p-r)+ \sum_{m \neq k,m\neq l} \Theta_{mm} r(1-r)) \\
    [(B \Theta \bar{B}) \cdot B]_{kk} &= (p+r)(\Theta_{kk} (p+r)(1-p-r)+\sum_{m \neq k} \Theta_{mm} r(1-r)).
  \end{align*} 
  With the above, we rewrite \eqref{lem:diag:ineq1} as follows:
  \begin{align}
    \frac{r ( \Theta_{kk} (p+r)(1-r) + \Theta_{ll} r(1-p-r)+ \sum_{\substack{m\neq k\\ m\neq l}} \Theta_{mm} r(1-r))} 
    {(p+r)(\Theta_{kk} (p+r)(1-p-r)+\sum_{m \neq k} \Theta_{mm} r(1-r))}   
    &> \frac{r}{p+r} \label{pr:diageq21} \\
    p(\Theta_{kk}(p+r) - \Theta_{ll}r)  &> 0, \label{pr:diageq22}
  \end{align} 
	where \eqref{pr:diageq21} to \eqref{pr:diageq22} follows from algebraic
	manipulation. Note that \eqref{pr:diageq22} is always true because of the
	lemma assumptions. In addition, by going through the same procedure, it can
	be shown that  
  \[ 
    \frac{[(B\Theta \bar{B} + \bar{B}\Theta B )\cdot B]_{kl}}
    {[(B\Theta \bar{B}+\bar{B}\Theta B)\cdot B]_{kk}} 
    > \frac{B_{kl}}{B_{kk}}. 
  \]
  In terms of the expected adjacency matrices, the above statement is equivalent
  to the following result. Suppose that nodes $k$ and $m$ belong to the same block
  and $l$ belongs to a different block, then
  \begin{equation} \label{lem:diag:ineq4}
    \frac{\tilde \W_{kl}}{\tilde \W_{km}} > \frac{\A_{kl}}{\A_{km}}.
  \end{equation}
  Now, we show $P_{\tilde \W}(u,v) < P_{\A}(u,v)$, when $u$ and $v$ belong to the same block. We have
  \begin{align*}
    \sum_{w\in \V} P_{\tilde \W}(u,w) &= \sum_{w\in \V} P_{\A}(u,w) = 1 \\
    \sum_{w\in \V} [\T^{-1} \tilde \W]_{uw} &= \sum_{w\in \V} [\D^{-1} \A]_{uw}.
  \end{align*} 	

  Assume $u$ and $v$ belong to block $\C$ of size $|\C|$. Factor out the
  transition probability between $u$ and $v$. Then,
  \[
    [\T^{-1} \tilde \W]_{uv} \left(|\C| + \sum_{w \notin \C} 
    \frac{[\T^{-1} \tilde \W]_{uw}}{[\T^{-1} \tilde \W]_{uv}} \right) 
    = [\D^{-1} \A]_{uv} \left( |\C| + \sum_{w \notin \C} 
    \frac{[\D^{-1} \A]_{uw}}{[\D^{-1} \A]_{uv}} \right).
  \]
 Since the summations are along the rows, we have
  \[ 
    [\T^{-1} \tilde \W]_{uv} \left(|\C| + \sum_{w \notin \C} \frac{\tilde \W_{uw}}{\tilde \W_{uv}} \right) 
    = [\D^{-1} \A]_{uv} \left(|\C| + \sum_{w \notin \C } \frac{\A_{uw}}{\A_{uv}} \right).
  \]
  Therefore, based on inequality \eqref{lem:diag:ineq4}, 
  \[ [\T^{-1} \tilde \W]_{uv} < [\D^{-1} \A]_{uv}. \]
  
  Now consider the case where $\Theta_{kk} = \Theta_{ll}$ for all $k$ and $l$, then for $w \notin \C$ 
  \begin{align} \label{lem:ineq:last}
    [\T^{-1} \tilde \W]_{uw} > [\D^{-1} \A]_{uw}. 
  \end{align}
\end{proof}

\begin{proof}[Proof of Proposition \ref{lem:spgap_popgraph}] 
The first part of this proof focuses on the inequality
$\lambda_2(P_{\tilde \W}) < \lambda_2(P_{\A})$. 
To this end, define $\B^{RW} := \D^{-1}_{k\times k} B,$ 
and $\B^{AC} := \T^{-1}_{k\times k} [(B\Theta \bar{B} + \bar{B}\Theta B)\cdot B]$. 
Since $\Theta_{kk}$ are all equal, then $\B^{RW}$ and $\B^{AC}$ 
are symmetric matrices and have equal row sum. Hence,
  \begin{align*}
    \lambda_2(P_{\A}) &= \lambda_2(\D^{-1} \A ) = \lambda_2(\B^{RW}), \\
    \lambda_2(P_{\tilde \W}) &= \lambda_2(\T^{-1} \tilde \W) = \lambda_2(\B^{AC}).
  \end{align*}		 

  Let $f:\{1,2,\cdots ,k\} \rightarrow \R$ and $r$ be the row sum of
  $\B^{RW}$ and $\B^{AC}$. Then  $I - \frac 1r \B^{AC}$ and $I - \frac 1r
  \B^{RW}$ are Laplacian matrices. Therefore,
  \begin{align*}
    \lambda_2( I - \frac 1r\B^{AC}) &= \inf_{ \substack{f:\sum_{u} f(u) = 0 \\ 
     f:\sum_u f^2(u) = 1} } \frac{1}{2r} {\sum_{u,v\;u\neq v} \B_{uv}^{AC} (f(v) - f(u))^2}\\
    &>  \inf_{\substack{f:\sum_{u} f(u) = 0 \\ 
  f:\sum_u f^2(u) = 1} } \frac{1}{2r} \sum_{u,v\;u\neq v} 
  \B_{uv}^{RW} (f(v) - f(u))^2  = \lambda_2( I - \frac 1r\B^{RW}),
  \end{align*}
  where the inequality follows from inequality \eqref{lem:ineq:last} and the fact
  that  $\B^{AC}_{uv} > \B^{RW}_{uv}$ for $u \neq v$. So we conclude that
  \[
  \lambda_2(\B^{AC}) < \lambda_2(\B^{RW})
  \] 
   and, therefore 
   \[
   \lambda_2(P_{\tilde \W}) < \lambda_2(P_{\A}).
   \]
This result is extended in the calculations below. 

The fact that $\lambda_2(P_{\A}) = 1/(R+1)$  follows immediately from Lemma \ref{fact:lam2}.

The rest of the proof is dedicated to equation \eqref{eq:smallworld} in the statement of the proposition.  From
Lemma \ref{lem:ac_as_sbm}, $ \tilde \W = Z \tilde B Z^T$ for $\tilde B =  (B\Theta
\bar{B} + \bar{B}\Theta B)\cdot B$.  Define $r' = 1-r$. Note that $\Theta = N/K I$, so it can be temporarily ignored as a constant. 

We have,
\[ B \bar{B} = (r'J - pI)(rJ+pI) = (r'rK + r'p - pr)J - p^2I.\]
Now, define $u = (r'rK + r'p - pr)$ and write 
\[ (B \bar{B}) \cdot B = (uJ - p^2I) \cdot(rJ + pI)  = p(u - rp - p^2)I + ur J .\]
Reincorporating the constants from  $\Theta = N/K I$ and a 2 to account for
$\bar{B} B$, it follows that $\tilde B = \tilde p I + \tilde r J$, for
\[ \tilde p = 2p(N/K)(u - rp - p^2) \quad \mbox{and} \quad \tilde r = 2(N/K)ur. \]
Note that $\tilde r$ and $\tilde p$ depend on the block populations $N/K$ and
thus the number of nodes in the graph $N$.  However, this term cancels out in
the ratio $\tilde r/\tilde p$.  So, neither $\lambda_2(P_{\tilde \W})$ nor
$\lambda_2(P_{\A})$ depend on $N$.  As such, 
\[\lambda_2(P_{\tilde \W}) + \epsilon < \lambda_2(P_{\A}) \]
for some $\epsilon > 0$ that is independent of $N$.  

As $K$ grows and $r$ shrinks, $u \rightarrow p(R+1)$ and 
  \[ 
    \tilde p \rightarrow 2 p (N/K)(p(R+1) - p^2) \quad \mbox{and} \quad 
    \tilde r \rightarrow 2 rp (N/K) (R+1).
  \]

  Using Lemma \ref{fact:lam2} on $\tilde B$, 
  \[\lambda_2(P_{\tilde \W}) = \frac{1}{K(\tilde r/ \tilde p) +1}.\]
  Then, 
  \[
    \frac{K\tilde r}{\tilde p} \rightarrow 
    \frac{Krp(R+1)}{p(p(R+1) - p^2)} = 
    \frac{Kr(R+1)}{p(R+1 - p)} = 
    R \frac{R+1}{R+1 - p},
  \]
  which concludes the proof.
\end{proof}

\begin{proof}[Proof of Proposition \ref{prop:epsilon}]
  Both $P_{\A}$ and $P_{\tilde \W}$ satisfy the conditions of Lemma
  \ref{lem:lam2}. It is only necessary to compute the probabilities. For
  $P_{\A}$, $p = 1-\epsilon$ and $r = \epsilon$.  So, \[\lambda_2(P_{\A}) =
  \frac{1-2\epsilon}{1} \rightarrow 1.\] To compute $\lambda_2\left(P_{\tilde \W} \right)$,
  notice that it is only necessary to determine $p$ and $r$ up to
  proportionality.  Under the assumed model, $\bar B_{11} = \epsilon$, $\bar
  B_{12} = 1-\epsilon$, and $\Theta  \propto I$.  Moreover, the matrix $(B\Theta \bar
  B + \bar B \Theta  B) \cdot B$ contains the elements $p = 2(1-\epsilon)^2$ and
  $r = (1-\epsilon)^2 + \epsilon^2$ for $P_{\tilde \W}$.  
  By Lemma \ref{lem:lam2}.  
  \[
    \lambda_2(P_{\tilde \W}) = \frac{2(1-\epsilon)^2
    -(1-\epsilon)^2 + \epsilon^2}{2(1-\epsilon)^2 + (1-\epsilon)^2 + \epsilon^2}
    = \frac{(1-\epsilon)^2 +\epsilon^2}{3(1-\epsilon)^2 +\epsilon^2} \rightarrow
    1/3.
  \]
\end{proof}



\section{Sampled Graph Results} \label{app:sampled_results}

Here we provide the proofs of Theorems \ref{thm:concentration_ac_laplacian} 
and \ref{thm:var_acrds} from  Section \ref{subsec:SampleGraph}.

\begin{proof}[Proof of Theorem \ref{thm:concentration_ac_laplacian}]
  By Lemma \ref{lem:symrw}, and Weyl's inequality, 
  \begin{equation*}
    \left| \lambda_\ell(P_{\tilde W}) - \lambda_\ell(P_{\tilde \W}) \right| =
    \left| \lambda_\ell(T^{-\half} \tilde W T^{-\half}) - \lambda_\ell(\T^{-\half}
    \tilde \W \T^{-\half}) \right| 
    \le 
    \left \| T^{-\half} \tilde W T^{-\half} - \T^{-\half} \tilde \W \T^{-\half} \right \|.
  \end{equation*}
  The rest of the proof studies the right hand side of this inequality.  

  For convenience and compactness, we introduce the following notation, 
  \begin{align*}
    N &:= |\V|, \\
    \tA       &:= (A\bA + \bA A), \\
    \tAc      &:= (\A\bAc + \bAc\A), \\
    \tilde W  &:= (A\bA+\bA A)\cdot A = \tA\cdot A,\\
    \tilde \W &:= (\A \bAc + \bAc\A)\cdot \A = \tAc \cdot \A. 
  \end{align*}

  By the triangle inequality,
  \begin{equation*}
    \| T^{-\half} \tilde W T^{-\half} - \T^{-\half} \tilde \W \T^{-\half} \|
    \leq \| \T^{-\half}(\tilde W - \tilde \W)\T^{-\half} \|  
    + \|T^{-\half} \tilde W T^{-\half} - \T^{-\half} \tilde W \T^{-\half} \|.
  \end{equation*}
  Also,
  \begin{align*}
    \| \T^{-\half}(\tilde W - \tilde \W)\T^{-\half} \| 
       &= \| \T^{-\half}( \tA\cdot A - \tAc\cdot\A)\T^{-\half} \| \\
    &\leq \| \T^{-\half}((\tA - \tAc)\cdot A)\T^{-\half} \| 
       + \| \T^{-\half}((A-\A)\cdot \tAc)\T^{-\half} \|.  
   \end{align*}
   
The remainder of the proof is divided into four parts. The terms $\| \T^{-\half}((\tA - \tAc)\cdot A)\T^{-\half}\|$, $\|\T^{-\half}((A-\A)\cdot \tAc)\T^{-\half}\|$, and $\|T^{-\half} \tilde W T^{-\half} - \T^{-\half} \tilde W \T^{-\half} \|$ are bounded in Part 1,2, and 3, respectively. Finally, Part 4 combines these bounds and completes the argument.

 \noindent \textbf{Part 1.} Note that $\T$ is a diagonal matrix and  $\tA$ and $\tAc$ are both
  symmetric. Therefore, we apply Lemma \ref{lem:diag_norm_invariant} to obtain
  \begin{align}\label{equ:Thm1:P1:eq1}
    \| \T^{-\half}((\tA - \tAc) \cdot A)\T^{-\half} \| 
    &= \| \T^{-1}(\tA - \tAc)\cdot A \| \notag \\
    &\leq \| \T^{-1}(A\bA-\A\bAc)\cdot A \| 
      + \| \T^{-1}(\bA A-\bAc\A)\cdot A \|.
  \end{align} 
  
It is sufficient to prove an upper bound for the first term in \eqref{equ:Thm1:P1:eq1}. The same bound will hold for the second term. We have
  \begin{equation}
    \|\T^{-1} (A\bA  - \A \bar{\A}) \cdot A \|  
    \leq \| \T^{-1} |A\bA  - \A \bar{\A}| \cdot A \|, 
  \end{equation}
  where $|\cdot|$ is the element-wise absolute value operator. The inequality follows from the fact that for any matrix $M$,
  $\|M\|\leq \left\||M| \right \|$ \citep[e.g.][Theorem 2.5]{mathias1990spectral}.
  
  We begin by bounding  the row sums of $|A\bA - \A\bAc|\cdot A$ with a concentration
  inequality. Then we use Lemma \ref{lem:posspup} to bound the operator
  norm. Define the row sum mapping $r_i$, so that for a matrix $C$, $r_i(C)$ equals the sum of the $i^{th}$ row of $C$. We have
  \begin{align} \label{app:thm1:part1:ineq3}
    r_i\left(\T^{-1} |A\bA  - \A \bar{\A}| \cdot A\right) &= 
    \frac{1}{\T_{ii}} \sum_j  A_{ij} 
      \left| \sum_k A_{ik}\bA_{kj} - \A_{ik} \bar{\A}_{kj}\right|.
  \end{align} 
  
  Define $F_{ij}=\sum_k \A_{ik}\bAc_{kj}$ and $G_{ij}=\sum_k \bAc_{ik}\A_{kj}$.
  For fixed $i$ and $j$, the random variables $\left\{ A_{ik}\bA_{kj} \right\}_k$ are independent with
  expected value $\rmE [A_{ik}\bA_{kj}] = \A_{ik} \bar{\A}_{kj}$ and 
  variance 
  \begin{equation*}
    \sigma^2_{ij} = \sum_k \rmE (A_{ik}\bA_{kj} - \A_{ik}\bar{\A}_{kj})^2  
    \leq \sum_k \rmE (A_{ik}\bA_{kj})^2 + (\A_{ik}\bar{\A}_{kj})^2 
    \leq 2F_{ij}.
  \end{equation*}
  Let $\Delta_{F_{ij}} := \sqrt{10F_{ij}\ln\frac{2N^2}{\de}}$. 
  By Bernstein's Inequality  and the union bound, 
  \begin{align} \label{app:thm1:part1:bernseq_ineq}
    \Pr \left( \left| \sum_k A_{ik}\bA_{kj} - 
    \A_{ik} \bar \A_{kj} \right| 
    \geq \Delta_{F_{ij}} \right) 
    &\leq 2 \exp  \left(-\frac{\frac 12 \Delta_{F_{ij}}^2}
    {\sigma_i^2 + \frac 13 S\Delta_{F_{ij}}} \right) \\
    &= 2 \exp\left(-\frac{ 5F_{ij}\ln\frac{2N^2}{\de}}
    {4F_{ij}+\frac 13 S\Delta_{F_{ij}}} \right) \notag\\	
    &\leq \frac{\de}{N^2} , \notag
  \end{align}
  where the last inequality follows from the assumption that $F_{\min} \gg \ln N$.
  So, with high probability,
  \begin{align*}
    \sum_j A_{ij} 
    \left|\sum_k A_{ik}\bA_{kj} - \A_{ik} \bar{\A}_{kj}\right| 
     \leq \sum_j A_{ij} \Delta_{F_{ij}}.
  \end{align*}
  
  Now we bound $|\sum_j (A_{ij}-\A_{ij}) \Delta_{F_{ij}}|$. We have $\rmE[A_{ij} \Delta_{F_{ij}}]=\A_{ij}\Delta_{F_{ij}}$ and
  \begin{equation*} 
    \sum_j \rmE[A_{ij}\Delta_{F_{ij}}- \A_{ij}\Delta_{F_{ij}}]^2
    \leq 2\sum_j \A_{ij}\Delta^2_{F_{ij}}.
  \end{equation*}
  By Bernstein's Inequality, the following holds with high probability
  \begin{equation*}
    \left|\sum_j (A_{ij}-\A_{ij}) \Delta_{F_{ij}} \right|
    \leq \sqrt{2\sum_j \A_{ij}\Delta^2_{F_{ij}} }.
  \end{equation*}
   Consequently, 
   \begin{align} \label{app:thm1:part1:ineq4}
    \sum_j A_{ij} 
    \left|\sum_k A_{ik}\bA_{kj} - \A_{ik} \bar{\A}_{kj}\right| 
    &\leq \sum_j \A_{ij}\Delta_{F_{ij}}
    + \sqrt{2\sum_j \A_{ij}\Delta^2_{F_{ij}} } \\
    &\leq 2\sum_j \A_{ij}\sqrt{10F_{ij}\ln\frac{2N^2}{\de}} \notag \\
    &\leq 10\sum_j \A_{ij}\sqrt{F_{ij}\ln\frac N\de}.\notag
  \end{align}
  Furthermore, 
  \begin{align}\label{app:thm1:part1:ineq5}
    \T_{ii} = \sum_{j} \tilde \W_{ij} 
    = \sum_{j} \A_{ij} \sum_k \A_{ik}\bAc_{kj} + \bAc_{ik}\A_{kj}
    = \sum_{j} \A_{ij}(F_{ij} + G_{ij}). 
  \end{align}
  From \eqref{app:thm1:part1:ineq3}, \eqref{app:thm1:part1:ineq4} 
  and \eqref{app:thm1:part1:ineq5}, 
  \begin{align}
    r_i\left(\T^{-1} |A\bA  - \A \bar{\A}| \cdot A\right) 
    \leq \frac{10\sum_j \A_{ij} \sqrt{F_{ij}\ln\frac N\de} }
    {\sum_{j} \A_{ij}F_{ij}} 
    \leq \frac{10\sum_j \A_{ij} F_{ij}\sqrt{\frac{\ln\frac N\de}{F_{ij}}}}
    {\sum_{j} \A_{ij}F_{ij}} 
    \leq 10 \sqrt{\frac{\ln\frac N\de}{F_{\min}}}.
  \end{align}
  Following the same steps, we obtain
  \begin{align}
    r_i\left(\T^{-1} |\bA A  - \bAc\A| \cdot A\right) 
    \leq \frac{10\sum_j \A_{ij}\sqrt{G_{ij}\ln\frac N\de}}
    {\sum_{j} \A_{ij}G_{ij}} 
    \leq 10 \sqrt{\frac{\ln\frac N\de}{G_{\min}}}.
  \end{align}
  Therefore, 
  \begin{equation} \label{th:eq:1:fibound}
    \| \T^{-\half}((\tA-\tAc)\cdot A)\T^{-\half} \| 	
    \leq \frac{10\ln^\half\frac N\de}
    {\min\{F^\half_{\min},G^\half_{\min}\}}.
  \end{equation}	

\noindent \textbf{Part 2.} We have
{\small
  \begin{equation}\label{app:thm1:part2:ineq0}
    \| \T^{-\half}((A-\A)\cdot \tAc)\T^{-\half} \| 
    \leq \| \T^{-\half}((A-\A)\cdot (\A\bAc))\T^{-\half} \| 
      + \| \T^{-\half}((A-\A)\cdot (\bAc\A))\T^{-\half} \|. 
  \end{equation}}
Similar to Part 1, it is sufficient to prove an upper bound for the first term in \eqref{app:thm1:part2:ineq0}. The same bound will hold for the second term. 
  
Let $J$ be the $N \times N$ square matrix comprised of all ones. We have
{\small
  \begin{align}\label{app:thm1:part2:ineq1}
    \|\T^{-\half}((A-\A)\cdot \A \bAc) \T^{-\half}\| 
    &= \|\T^{-\half}((A-\A)\cdot \A (J-\A))\T^{-\half}\| \notag \\
    &= \|\T^{-\half}((A-\A) \cdot (\A J) - 
       (A-\A) \cdot \A\A) \T^{-\half}\| \notag\\
    &\leq \|\T^{-\half}\D(A-\A)\T^{-\half}\| 
     + \|\T^{-\half}((A-\A)\cdot \A\A)\T^{-\half}\| \notag\\
    &= \|\T^{-\half}\D^\half (A-\A)\D^\half\T^{-\half}\| 
     + \|\T^{-\half}((A-\A) \cdot \A\A)\T^{-\half}\|. 
  \end{align} }
Consider the first term in \eqref{app:thm1:part2:ineq1}. For $i,j \in \left\{ 1,\cdots,N \right\}$,  define $A^{ij}\in\{0,1\}^{N\times N}$ to be the matrix
with one at elements $(i,j)$ and $(j,i)$, and zero everywhere else. We then have, 
  \begin{align*}
    \T^{-\half} \D^{\half}(A-\A) \D^{\half} \T^{-\half} 
    =\sum_{i=1}^N \sum_{j>i}^N  \sqrt{\frac{\D_{ii} \D_{jj}}{\T_{ii} \T_{jj}}} 
    (A_{ij} - \A_{ij}) A^{ij}.
  \end{align*} The right hand side is a sum of independent, symmetric matrices.
  Therefore, we can apply Theorem 5 of \cite{chung2011spectra} to bound it. Let
\begin{equation*} 
M  := \max_{ij=1,\cdots,N} \left\| \sqrt{\frac{\D_{ii} \D_{jj}}{\T_{ii} \T_{jj}}}  (A_{ij} - \A_{ij}) A^{ij} \right\|    \leq \max_{ij=1,\cdots,N}  \sqrt{\frac{\D_{ii} \D_{jj}}{\T_{ii}\T_{jj}}},
\end{equation*} 
  and   
  \begin{align*}
    v^2 &:= \left\| \sum_{i=1}^N \sum_{j>i}^N \mathrm{Var} \left( \sqrt{\frac{\D_{ii}
      \D_{jj}}{\T_{ii} \T_{jj}} } (A_{ij} - \A_{ij} ) A^{ij} \right) \right\| \\ 
    &= \left\| \sum_{i=1}^N \sum_{j>i}^N \left[ \frac{\D_{ii} \D_{jj}}{\T_{ii}
      \T_{jj}}  (\A_{ij}-\A_{ij}^2) A^{ii} \right] \right\| \notag\\
    &\leq \max_{i=1,\cdots,N} \left( \sum_{j=1}^N \left[ \frac{\D_{ii}
       \D_{jj}}{\T_{ii} \T_{jj}}  (\A_{ij}-\A_{ij}^2) \right] \right) \notag\\
    &\leq \max_{i=1,\cdots,N} \left( \sum_{j=1}^N \left[ \frac{\D_{ii}
       \D_{jj}}{\T_{ii} \T_{jj}}  \A_{ij} \right] \right)
    \leq \max_{ij=1,\cdots,N} \frac{\D_{ii}^2 \D_{jj}} {\T_{ii}\T_{jj}}.\notag
  \end{align*}

  Define
  \begin{equation*}
    \Delta := \max_{ij=1,\cdots,N} 2\sqrt{\frac{\D_{ii}^2 \D_{jj} 
          \ln(2N/\de)} {\T_{ii}\T_{jj}}}.
  \end{equation*} 
  Note that
  \[
    M \Delta = \max_{ij} 
    \sqrt{\frac{\D_{ii} \D_{jj}}{\T_{ii}\T_{jj}}}
    \sqrt{\frac{\D_{ii}^2 \D_{jj} 
    \ln(2N/\de)} {\T_{ii}\T_{jj}}} 
    = \max_{ij} \frac{\D_{ii}^2\D_{jj}}{\T_{ii}\T_{jj}} 
    \sqrt{\frac{\ln (2N/\delta) }{\D_{ii}}}
    \leq v^2 \sqrt{\frac{\ln(2N/\delta)}{\D_{\min}}}.
  \]
  
  Therefore, applying Theorem 5 in \cite{chung2011spectra} yields
    \begin{align} \label{app:thm1:part2:ineq2}
    \Pr\left( \left\| \sum_{i=1}^N \sum_{j>i}^N 
      \sqrt{\frac{\D_{ii} \D_{jj}} {\T_{ii}\T_{jj}}} 
      (A_{ij}-\A_{ij}) A^{ij} \right\| 
      \geq \Delta\right) &\leq 2N \exp\left(-\frac{\Delta^2}{2v^2 + 2 M \Delta/3} \right) \\
    &\leq \de. \notag
  \end{align}

  For the second term of \eqref{app:thm1:part2:ineq1}, we obtain
  \begin{align*}
    \T^{-\half} \left( (A-\A)\cdot \A\A \right) \T^{-\half} 
    = \sum_{i=1}^N \sum_{j>i}^N \sqrt{ \frac{1} {\T_{ii}\T_{jj}}} 
    \left( A_{ij}-\A_{ij} \right) \left(\sum_{k=1}^N \A_{ik}\A_{kj} \right) A^{ij}.
  \end{align*} 
  Because $\left|\sum_k \A_{ik} \A_{kj} \right| \le \sqrt{\D_{ii} \D_{jj}}$, we obtain the
  same bound as \eqref{app:thm1:part2:ineq2}. Namely, 
  \begin{align*}
    \Pr\left( \left\| \sum_{i=1}^N \sum_{j>i}^N \sqrt{\frac{1}{\T_{ii}\T_{jj}}} 
    (A_{ij}-\A_{ij}) (\sum_{k=1}^N \A_{ik}\A_{kj}) A^{ij} \right\| \geq \Delta \right) 
    &\leq 2N \exp\left(-\frac{\Delta^2}{2v^2 + 2 M \Delta/3 }\right) \\
    &\leq \de.
  \end{align*}
  In addition,
  \begin{align*}
    \T_{ii} &= \sum_j \A_{ij}  \sum_k \A_{ik}\bAc_{kj}+ \bAc_{ik}\A_{kj}\\
            &= \sum_j \A_{ij} (F_{ij} + G_{ij})\\
            &\geq \sum_j\A_{ij} c_1 \D_{ii} 
            \geq c_1 \D_{ii}^2,
  \end{align*}
where the inequality follows from the assumption that $F_{ij} + G_{ij} > c_1 \D_{ii}$ for all $i,j \in \left\{ 1,\cdots,N \right\}$.

Combining the above results, yields 
  \begin{equation}\label{equ:thm1:part2:T1} 
    \|\T^{-\half} ((A-\A) \cdot \A \bAc)\T^{-\half} \| 
    \leq 4 \sqrt{\frac{\ln \frac N\de}{c_1 \D_{\min}}}.
  \end{equation}
 
 As noted above, the second term in \eqref{app:thm1:part2:ineq0} satisfies the same bound, so that
   \begin{equation}\label{equ:thm1:part2:T2} 
    \| \T^{-\half}((A-\A)\cdot (\bAc\A))\T^{-\half} \|
    \leq 4 \sqrt{\frac{\ln \frac N\de}{c_1 \D_{\min}}}.
  \end{equation}
 
  Combining \eqref{app:thm1:part2:ineq0}, \eqref{equ:thm1:part2:T1}, and \eqref{equ:thm1:part2:T2}, yields 
  \begin{equation}
    \| \T^{-\half}((A-\A)\cdot \tAc)\T^{-\half} \| 
     \leq 8\sqrt{\frac{\ln \frac N\de}{c_1 \D_{\min}}}.
  \end{equation}

\noindent \textbf{Part 3.} First we bound $\left|T_{ii}-\T_{ii} \right|$ and then we bound
  $\left\| \T^{-\half} T^{+\half} -I \right\|$. 
  
  We have
  \begin{align}\label{equ:thm1:P3A:1} 
    \left|T_{ii}-\T_{ii} \right| &= \left| r_i(\tA\cdot A) - r_i(\tAc\cdot \A) \right| \notag \\
     &\leq \left|r_i((A \bA)\cdot A) - r_i((\A\bAc)\cdot \A) \right| 
     + \left|r_i((\bA A)\cdot A) - r_i((\bAc\A)\cdot \A) \right|.
  \end{align}
  Consider the first term in \eqref{equ:thm1:P3A:1},
  \begin{align} \label{app:thm1:part3:ineq0}
    \left|r_i((A\bA)\cdot A) - r_i((\A\bAc)\cdot \A) \right| 
    &= \left|\sum_j A_{ij} \sum_k A_{ik}\bA_{kj} - \sum_j \A_{ij} \sum_k \A_{ik} \bAc_{kj} \right| \\
    &\leq \sum_j A_{ij} \left| \sum_k A_{ik}\bA_{kj} - \A_{ik}\bar{\A}_{kj} \right|
    + \left| \sum_j (A_{ij} - \A_{ij}) \sum_k \A_{ik} \bAc_{kj} \right|. \notag
  \end{align} 

  To bound the first term of \eqref{app:thm1:part3:ineq0},
  we use \eqref{app:thm1:part1:bernseq_ineq} and \eqref{app:thm1:part1:ineq4}. With probability at least
  $1-\de$,
  \begin{equation} \label{app:thm1:part3:ineq1}
    \sum_j A_{ij} \left| \sum_k A_{ik}\bA_{kj} - \A_{ik}\bar{\A}_{kj} \right|  
    \leq 10\sum_j\A_{ij}\sqrt{F_{ij}\ln\frac N\de}.
  \end{equation}
  
  Consider the second term in \eqref{app:thm1:part3:ineq0},
  \begin{align*}
    \left| \sum_j(A_{ij}-\A_{ij}) \sum_k\A_{ik} \bAc_{kj} \right|
     &= \left| \sum_k \sum_j (A_{ij}-\A_{ij})\A_{ik} \bAc_{kj} \right| \\ 
     &\leq \sum_k \A_{ik} \left|\sum_j A_{ij}\bAc_{kj} - \A_{ij}\bAc_{kj} \right| \\ 
     &= \sum_k \A_{ik} \left| \sum_j A_{ij}\bAc_{jk} - \A_{ij}\bAc_{jk} \right|\\
     &= \sum_j \A_{ij} \left| \sum_k A_{ik}\bAc_{kj} - \A_{ik}\bAc_{kj} \right|. 
  \end{align*}
  
Note that,  $\rmE[A_{ik}\bAc_{kj}]= \A_{ik}\bAc_{kj}$. In addition, we can obtain the same upper bound for the variance to use \eqref{app:thm1:part1:bernseq_ineq}. Hence, with probability at least $1-\de$, we have
  \begin{equation} \label{app:thm1:part3:ineq2}
    \sum_j \A_{ij} \left| \sum_k A_{ik}\bAc_{kj} - \A_{ik}\bar{\A}_{kj} \right| 
    \leq 10\sum_j\A_{ij}\sqrt{F_{ij}\ln\frac N\de}.
  \end{equation}
                                             
  From \eqref{app:thm1:part3:ineq1} and \eqref{app:thm1:part3:ineq2}, we have
  \begin{equation}
    \left| r_i\left((A\bA)\cdot A \right) - r_i\left((\A\bAc)\cdot \A \right) \right| 
    \leq 20\sum_j \A_{ij}\sqrt{F_{ij}\ln\frac N\de}.
  \end{equation} 
  For the second term in \eqref{equ:thm1:P3A:1}, following the same steps yields 
  \begin{equation}
    \left| r_i \left((\bA A)\cdot A \right) - r_i \left((\bAc\A)\cdot \A \right) \right| 
    \leq 20\sum_j \A_{ij} \sqrt{G_{ij}\ln\frac N\de}.
  \end{equation} 
  Therefore,
  \begin{align*}
    \left |T_{ii} - \T_{ii} \right| \leq 40\sum_j \A_{ij}
    \left(\sqrt{F_{ij}\ln\frac N\de} + \sqrt{G_{ij}\ln\frac N\de} \right).
  \end{align*}		

  Now we consider $\| \T^{-\half} T^{+\half} -I\|$. We have
  \begin{align*}
    \left\| \T^{-\half} T^{+\half} -I \right \| &\leq 
    \max_{i=1,\cdots,N}\left| \sqrt \frac{T_{ii}}{\T_{ii}}-1 \right| \\
    &\leq \max_{i=1,\cdots,N}\left| \frac{T_{ii}}{\T_{ii}}-1 \right| \\
    &\leq \max_{i=1,\cdots,N} \frac{40\sum_j \A_{ij}
    \left( \sqrt{F_{ij}\ln\frac N\de} + \sqrt{G_{ij}\ln\frac N\de} \right)}
    {\sum_j \A_{ij}(F_{ij}+G_{ij})} \leq \frac{40\ln^\half\frac N\de}
    {\min\{G^\half_{\min},F^\half_{\min}\}}.
  \end{align*}
  Furthermore, 
  \begin{equation}
   \left \| \T^{-\half} T^{+\half} \right \| \leq 1 +
    \frac{40\ln^\half\frac N\de}{\min\{G^\half_{\min},F^\half_{\min}\}}
    <2,
  \end{equation}
   where the last inequality follows from the theorem's assumptions.

  Define the Laplacian matrix $L^{ac} := T^{-\half} \tilde W T^{-\half}$. So,
  \begin{align*}
      \| T^{-\half} \tilde W T^{-\half} - \T^{-\half} \tilde W \T^{-\half} \| 
    &= \| T^{-\half} \tilde W T^{-\half} - \T^{-\half}T^{+\half}T^{-\half}  
       \tilde W T^{-\half}T^{+\half}\T^{-\half} \| \\
    &= \| I-L^{ac} - \T^{-\half} T^{+\half} \{I - L^{ac}\} T^{+\half}T^{-\half} \| \\
    &= \| \{\T^{-\half} T^{+\half} -I\}\{I-L^{ac}\}T^{+\half} \T^{-\half} + 
        \{I-L^{ac}\} \{I-T^{+\half}\T^{-\half}\}\| \\
    &\leq  \| \T^{-\half} T^{+\half} -I\| \cdot \| T^{+\half} \T^{-\half}\| + 
    \| I - T^{+\half} \T^{-\half}\|, 
  \end{align*}
  where the inequality follows from the fact that $\|I-L^{ac}\| \leq 1$. 
  Now,
  \begin{align*}
    \| T^{-\half} \tilde W T^{-\half} - \T^{-\half} \tilde W \T^{-\half} \| 
    \leq 
    \frac{120\ln^\half\frac N\de}{\min\{G^\half_{\min},F^\half_{\min}\}}. 
  \end{align*}
 %

  \noindent \textbf{Part 4.} Let $\e := 10\de$. Combining the results of the three preceding parts yields 
  \begin{equation*}
    \left \| T^{-\half} \tilde W T^{-\half} - \T^{-\half} \tilde \W \T^{-\half} \right \|
    \leq
      \frac{10\ln^\half\frac N\de}{\min\{G^\half_{\min},F^\half_{\min}\}}
      + \frac{8\ln^\half\frac N\de}{c_1^\half \D^\half_{\min}} 
    + \frac{120\ln^\half\frac N\de}{\min\{G^\half_{\min},F^\half_{\min}\}}.
  \end{equation*} 
  Note that $G_{\min} = F_{\min}$ and $D_{\min}\geq F_{\min}$. So
  \begin{equation*}
    \left \| T^{-\half} \tilde W T^{-\half} - \T^{-\half} \tilde \W \T^{-\half} \right \|
    \leq \frac{138\ln^\half\frac {10N}{\e}}{\min\{c_1^\half \D^\half_{\min}, 
      F^\half_{\min}\}}
    \leq \frac{138\ln^\half\frac {10N}{\e}}{c_1^\half F^\half_{\min}},
  \end{equation*} 
  with probability at least $1-\e$. 
\end{proof}


\begin{proof}[Proof of Theorem \ref{thm:var_acrds}]
Define ${f_j}$ as the $j^{th}$ eigenvector of $P_A$ with respect to the inner product $\inner{\cdot}{\cdot}_{\pi}$; similarly, define ${f^{ac}_j}$ as the $j^{th}$ eigenvector of $P_{\tilde W}$ with respect to the inner product $\inner{\cdot}{\cdot}_{\pi^{ac}}$. From Proposition \ref{lem:cov_eq}, to prove the theorem, it is sufficient to show that
  \[ 
    \sum_{j=2}^{|\V|}  \inner{y}{f^{ac}_j}_{\pi^{ac}}^2 \lambda_j(P_{\tilde W})^t  <  
    \sum_{j=2}^{|\V|}  \inner{y}{f_j}_{\pi}^2 \lambda_j(P_A)^t.  
  \]
We break the proof into two steps. In the first step, we show that  the above holds true in the population, i.e. we compare the Markov chains on $P_{\tilde \W}$ and $P_\A$.  In the second step, we show that the sample quantities converge almost surely to the population quantities. 

\noindent \textbf{Part 1.} In this step, we show that
  \[ 
    \sum_{j=2}^{|\V|}  \inner{y}{\bar f^{ac}_j}_{\bar \pi}^2
    \lambda_j(P_{\tilde \W})^t  +\epsilon <  \sum_{j=2}^{|\V|}
    \inner{y}{\bar f_j}_{\bar\pi}^2 \lambda_j(P_\A)^t.  
  \]

We begin by analyzing the eigenpairs of the transition matrices. From Lemma \ref{fact:lam2}, for $i = 2, \dots, K$,
  \begin{equation}
    \lambda_i(P_\A) = \lambda_2(P_\A) \quad \mbox{ and } \quad
    \lambda_i(P_{\tilde \W} ) = \lambda_2(P_{\tilde \W}).
  \end{equation}
  Moreover, for $i>K,  \lambda_i(P_\A) =  \lambda_i(P_{\tilde \W} ) = 0$.
  Under the theorem conditions, $P_{\tilde
  \W}$ and $P_\A$ have the same stationary distribution; refer to this as
  $\bar\pi$ (in fact, this distribution is uniform on the nodes). Define
  $\bar f_j$ and $\bar f^{ac}_j$ as the $j$th eigenvectors, with respect to $\inner{\cdot}{\cdot}_{\bar\pi}$, of $P_\A$ and
  $P_{\tilde \W}$, respectively. Therefore, we have
  \[
    \sum_{j=2}^{|\V|}  \inner{y}{\bar f_j}_{\bar\pi}^2
    \lambda_j(P_\A)^t =   \lambda_2(P_\A)^t \sum_{j=2}^K
    \inner{y}{\bar f_j}_{\bar\pi}^2,
  \] 
  and 
  \[
    \sum_{j=2}^{|\V|} \inner{y}{\bar f^{ac}_j}_{\bar\pi}^2 \lambda_j(P_{\tilde \W})^t =
    \lambda_2(P_{\tilde \W})^t \sum_{j=2}^K \inner{y}{\bar f^{ac}_j}_{\bar\pi}^2.
  \] 
  Proposition \ref{lem:spgap_popgraph} shows thats  $\lambda_2(P_{\tilde \W})^t +\epsilon  < \lambda_2(P_{\A})^t $, where
  $\epsilon$ does not change asymptotically as $|\V|$ grows.  Thus, Part 1 will be finished after showing that $\sum_{j=2}^K
  \inner{y}{\bar f^{ac}_j}_{\bar\pi}^2 = \sum_{j=2}^K
  \inner{y}{\bar f_j}_{\bar\pi}^2$.  To compare these terms, note that the
  construction of the eigenvalues in the proof of Lemma \ref{fact:lam2} shows
  that the span of the sets $\{\bar f^{ac}_j\cdot\bar\pi^\half: j = 1, \dots, K\}$ and $\{\bar f_j\cdot\bar\pi^\half: j = 1, \dots, K\}$ are identical. 
  Therefore, Parseval's Identity implies,
  \[ 
    \sum_{j=1}^K \inner{y}{\bar f^{ac}_j}_{\bar\pi}^2 = \sum_{j=1}^K
    \inner{y}{\bar f_j}_{\bar\pi}^2.
  \]
  Note that $\bar f^{ac}_1=\bar f_1=\1$ because these are the lead eigenvectors of Markov transition matrices. Thus, 
  \[ 
    \sum_{j=2}^K \inner{y}{\bar f^{ac}_j}_{\bar\pi}^2 = \sum_{j=2}^K
    \inner{y}{\bar f_j}_{\bar \pi}^2.
  \] 
 
\noindent \textbf{Part 2.} To ease notation, let $\lambda_j :=\lambda_j(P_A)$ 
  and $\bar{\lambda}_j:=\lambda_j(P_\A)$. Finally, let $N := |\V|$ denote the size of the graph. This part of the proof shows that, as $N \rightarrow \infty$,
  \[
    \left| \sum_{j=2}^{|\V|} \inner{y}{\bar f_j}_{\bar\pi}^2 \bar{\lambda}_j^t 
    - \inner{y}{f_j}_{\pi}^2 \lambda_j^t \right|  \xra{as} 0.
  \]
The corresponding proof for the anti-cluster random walk follows from a similar argument.  
  
  We have
  \begin{align}
    \left| \sum_{j=2}^{|\V|} \inner{y}{\bar f_j}_{\bar\pi}^2 \bar{\lambda}_j^t 
    - \inner{y}{f_j}_{\pi}^2 \lambda_j^t \right| 
    &= \left| \sum_{j=2}^{|\V|} \inner{y}{\bar f_j}_{\bar\pi}^2
    \bar{\lambda}_j^t - \inner{y}{f_j}_{\pi}^2 (\bar{\lambda}_j^t
    + (\lambda_j^t-\bar{\lambda}_j^t)) \right| \notag \\
    &\leq \left| \sum_{j=2}^K \bar{\lambda}_j^t \left(
    \inner{y}{\bar{f}_j}_{\bar{\pi}}^2 - \inner{y}{f_j}_{\pi}^2 \right) \right| +
    \left| \sum_{j=2}^{|\V|} \inner{y}{f_j}_{\pi}^2 \left|\lambda_j^t-\bar{\lambda}_j^t\right| \right| \notag \\
    &\leq \bar{\lambda}_2^t \cdot \left|\sum_{j=2}^K \inner{y}{f_j}^2_{\pi}
    - \inner{y}{\bar f_j}^2_{\bar{\pi}}\right| + \max_j
    \left|\lambda_j^t-\bar{\lambda}_j^t\right| \cdot \inner{y}{y}_{\pi}^2. \label{eq:twoparts}
  \end{align}

Since $y$ is a bounded function, $\inner{y}{y}_{\pi}^2$ is bounded. Therefore, Theorem 1, with $\epsilon=1/N^2$, and the Borel-Cantelli Lemma imply 
that the second term in \eqref{eq:twoparts} converges to zero almost surely.

  Now we argue that the first term in \eqref{eq:twoparts} converges to zero. Let $\cdot$ denote element-wise
  multiplication.  Let $\pi^\half$ denote the vector with elements $\sqrt{\pi_i}$. Finally, 
  let $diag(\pi^\half)$ denote the diagonal matrix with $\pi^\half$ on the diagonal.
  For some constant $c$, 
  \begin{equation}\label{equ:Thm2:P2:a}
    \D^{-\half} \A \D^{-\half} (\bar f_j \cdot \bar\pi^\half)=
    \D^{-\half} \A \D^{-\half} diag(\bar\pi^\half) \bar f_j  =
    \D^{-\half} \A cI \bar f_j  = c \bar \lambda_j f_j.
  \end{equation}
 Note that $\inner{\bar f_j \cdot \bar\pi^\half}{\bar f_i
  \cdot \bar\pi^\half} \in\{0,1\}$; it is equal to one if and only if $i=j$. This fact combined with \eqref{equ:Thm2:P2:a}
  shows that  $\bar f_j \cdot \bar\pi^\half$ forms an orthonormal basis of the
  eigenspace of $\D^{-\half} \A \D^{-\half}$.  Similarly, this holds for $f_j \cdot
  \pi^\half$ and  $D^{-\half} A D^{-\half}$.

  Let $\bar V \in \R^{N \times (K-1)}$ and $V \in \R^{N \times (K-1)}$ 
  be matrices with columns defined by
  $\bar V_{j}:=\bar f_{j+1} \cdot \bar\pi^\half$ and
  $V_{j}:=f_{j+1} \cdot \pi^\half$, respectively, for $j \in \{ 1,\cdots,K-1\}$. 
  Note that the columns of $V$ and $\bar V$ are orthonormal. 
  Furthermore, define the corresponding orthogonal projection
  matrices $\bar Q=\bar V\bar V^T \in \R^{N\times N}$ 
  and $Q=VV^T \in \R^{N\times N}$. We then have
  \begin{align}\label{equ:Thm2:P2:b}
    \left| \sum_{j=2}^K \inner{y}{f_j}^2_{\pi}
    - \inner{y}{\bar f_j}^2_{\bar{\pi}} \right| 
    &=\left| \sum_{j=2}^K \inner{y\cdot\pi^\half}{f_j\cdot\pi^\half}^2
    - \inner{y\cdot\bar\pi^\half}{\bar f_j\cdot\bar\pi^\half}^2 \right| \notag \\
    &= \left| \left\|V^T \left(y\cdot\pi^\half\right) \right\|^2 
    - \left\| \bar V^T \left(y\cdot \bar\pi^\half\right) \right\|^2 \right| \notag \\
    &= \left| \left\|Q \left(y\cdot\pi^\half\right) \right\|^2 
    - \left\| \bar Q \left(y\cdot \bar\pi^\half\right) \right\|^2 \right| \notag \\
    &\leq \left\| Q \left( y\cdot\pi^\half \right) 
    - \bar Q \left( y\cdot\bar\pi^\half \right) \right\|^2 \notag \\
    &\leq \left\| \left(Q-\bar Q\right) \left( y\cdot\pi^\half \right) \right\|^2
    + \left\|\bar Q \left(y\cdot\pi^\half - y\cdot\bar\pi^\half \right) \right\|^2 \notag \\
    &\leq \left\| Q-\bar Q \right\|^2 \cdot \inner{y}{y}_\pi
    + \left\|\bar Q \left( y \cdot \left(\pi^\half - \bar\pi^\half\right) \right) \right\|^2.
  \end{align}
  Consider the first term in \eqref{equ:Thm2:P2:b}. Recall that $y$ is a bounded function, hence, it is sufficient to prove
  $$
  \left\| Q-\bar Q \right\| \xra{as} 0. 
  $$
  Define
  \[
    \delta = \min\left\{|\bar\lambda_{K+1}-\lambda_K|,|\bar\lambda_1-\lambda_2|\right\}.
  \] 
  From the Davis-Kahan Theorem \citep[e.g.][Theorem 1]{yu2015useful}, 
  it follows that
  \begin{equation} \label{thm2:part2:dkineq}
    \frac{\left\| D^{-\half} A D^{-\half} - \D^{-\half} \A \D^{-\half} \right\|}{\delta}
    \geq \left\| \sin\Theta\left(V, \bar V \right) \right\| = \left\| Q - \bar Q \right\|, 
  \end{equation}
  where the equality follows from \citet[][Theorem 5.5 pp. 43]{stewart1990matrix}. 
  Recall, $\bar\lambda_{K+1}=0$, $\bar\lambda_1=1$, and $\lambda_2 \in (0,1)$. Thus,
  $\delta=\min\left\{\lambda_K,1-\lambda_2\right\}$. Furthermore,
  $$
  \lambda_K> \left| \bar\lambda_K- |\lambda_K-\bar\lambda_K|  \right| \quad \text{and} \quad \lambda_2 > \left| \bar\lambda_2-|\lambda_2-\bar\lambda_2| \right|.
  $$
  Additionally, recall that $\bar\lambda_2=\bar\lambda_K$. 
  Then, Theorem  \ref{thm:concentration_ac_laplacian} implies
  $|\lambda_j-\bar\lambda_j| \xra{as} 0$, which is less than $\bar \lambda_2$.
  So,  $\delta>\frac 12\bar\lambda_2$. Theorem
  \ref{thm:concentration_ac_laplacian} also implies that the numerator on the
  left hand side of \eqref{thm2:part2:dkineq} converges almost surely to zero. Therefore,
  $$
  \left\| Q-\bar Q \right\| \xra{as} 0. 
  $$

  Now, consider the second term in \eqref{equ:Thm2:P2:b}. We have
  \begin{equation}
    \left\| \bar Q^T \left( y \cdot \left(\pi^\half - \bar\pi^\half \right) \right) \right\|^2
    \leq \left\| y \cdot \left(\pi^\half - \bar\pi^\half \right) \right\|^2
    \leq \|y\|_{\infty}^2 \cdot \left\| \pi^\half - \bar\pi^\half \right\|^2.
  \end{equation} 
It follows from the theorem assumptions that $\|y\|_{\infty}$ is bounded. Hence, it is sufficient to prove 
$$
\left\|\pi^\half - \bar\pi^\half \right\| \xra{as} 0.
$$
Note that $\left\| \pi^\half \right\| = \left\| \bar\pi^\half \right\| = 1$. So, 
\[ 
    \left\|\pi^\half - \bar\pi^\half \right\| 
    = 2\sin{\frac{\Theta(\pi^\half, \bar\pi^\half)}{2}}.
  \]  
  Recall that $\pi^\half$ and $\bar\pi^\half$ are leading eigenvectors of the sample
  and population Laplacian matrices, respectively. Then, it follows from the
  Davis-Kahan Theorem and concentration of eigenvalues of the Laplacian
  matrices that $\left|\sin\Theta\left(\pi^\half,\bar\pi^\half\right)\right| \xra{as} 0$.
  Therefore, we conclude that  
  $\left\|\pi^\half - \bar\pi^\half \right\| \xra{as} 0$. 

\end{proof}  


\end{document}